\documentclass[a4paper,10pt]{article}
\usepackage[a4paper, left=1.3in, right=1.5in, top=0.9in]{geometry}

\usepackage{amsmath,amssymb,mathtools}
\usepackage{soul}
\usepackage{graphicx,color}
\usepackage{xspace}
\usepackage{todonotes}
\usepackage{csquotes}
\usepackage{booktabs} 
\usepackage{tabularx} 
\usepackage{xfrac} 

\usepackage{tikz,authblk}

\usepackage{todonotes}

\usepackage[small]{complexity}
\usepackage[T1]{fontenc} 

\newclass{\paraNP}{paraNP}

\def\R{\mathbb R} 

\usepackage{microtype}

\usepackage{hyphenat}

\usepackage{xpunctuate}
\usepackage[all,british]{foreign} 
\redefnotforeign[ie]{i.e\xperiodafter}
\redefnotforeign[eg]{e.g\xperiodafter}

\usepackage{tikz}

\def\dist{\operatorname{d}}
\def\cP{\mathcal P}
\def\cI{\mathcal I}
\def\cY{\mathcal Y}
\def\Z{\mathbb Z}

\usepackage{stmaryrd}

\usepackage{accents}
\newlength{\dhatheight}

\renewcommand{\leq}{\leqslant}

\renewcommand{\geq}{\geqslant}

\renewcommand{\epsilon}{\varepsilon}
\newcommand{\eps}{\varepsilon}

\def\symdiff{\mathop{\vartriangle}}

\def\angled#1{\langle #1 \rangle}
\DeclareMathOperator{\opt}{\mathbf{opt}}

\DeclareMathOperator{\pack}{pack}
\DeclareMathOperator{\cost}{\mathbf{cost}}
\DeclareMathOperator{\ar}{\mathbf{ar}}

\DeclareMathOperator{\midup}{\mathbin{\curvearrowleft}}
\DeclareMathOperator{\middown}{\mathbin{\curvearrowright}}
\DeclareMathOperator{\lca}{\mathbin{\uparrow}}
\DeclareMathOperator{\lcaskew}{\mathbin{\nearrow}}

\DeclareMathOperator{\closest}{C_{min}}
\DeclareMathOperator{\furthest}{C_{max}}

\newenvironment{tightcenter}
 {\parskip=0pt\par\nopagebreak\centering}
 {\par\noindent\ignorespacesafterend}

\usepackage{tikz}
\usetikzlibrary{calc}
\usepackage{xargs}
\usepackage{xifthen}

\usepackage{framed}

\usepackage{ctable}
\newlength{\RoundedBoxWidth}
\newsavebox{\GrayRoundedBox}
\newenvironment{GrayBox}[1]%
   {\setlength{\RoundedBoxWidth}{\textwidth-4.5ex}
    \def\boxheading{#1}
    \begin{lrbox}{\GrayRoundedBox}
       \begin{minipage}{\RoundedBoxWidth}%
   }{%
       \end{minipage}
    \end{lrbox}%
    \begin{tightcenter}%
    \begin{tikzpicture}%
       \node(Text)[draw=black!20,fill=white,rounded corners,%
             inner sep=2ex,text width=\RoundedBoxWidth]%
             {\usebox{\GrayRoundedBox}};
        \coordinate(x) at (current bounding box.north west);
        \node [draw=white,rectangle,inner sep=3pt,anchor=north west,fill=white] 
        at ($(x)+(6pt,.75em)$) {\boxheading};
    \end{tikzpicture}
    \end{tightcenter}\vspace{0pt}%
    \ignorespacesafterend
}    

\newenvironment{problem}[2][]{\noindent\ignorespaces%
                                \FrameSep=6pt%
                                \parindent=0pt%
                \vspace*{-.5em}
                \ifthenelse{\isempty{#1}}{%
                  \begin{GrayBox}{\textsc{#2}}%
                }{%
                  \begin{GrayBox}{\textsc{#2} parametrised by~{#1}}%
                }
                \newcommand\Prob{Problem:}%
                \newcommand\Input{Input:}%
                \begin{tabular*}{\textwidth}{@{\hspace{.1em}} >{\itshape} p{1.6cm} p{0.8\textwidth} @{}}%
            }{
                \end{tabular*}%
                \end{GrayBox}%
                \vspace*{-.5em}
                \ignorespacesafterend
            }

\newcommand{\Problem}[1]{\textsc{#1}}

\usepackage{amsthm, thm-restate}
\usepackage{thm-restate}

\newtheorem{theorem}{Theorem}[section]
\newtheorem{lemma}[theorem]{Lemma}
\newtheorem{definition}[theorem]{Definition}
\newtheorem{corollary}[theorem]{Corollary}
\newtheorem{proposition}[theorem]{Proposition}
\newtheorem*{claim}{Claim}

\def\Wahlstrom{Wahlstr{\"o}m\xspace}

\def\plog{\log^{\kern-.1pt{\scriptscriptstyle O(1)}}\kern-2pt}

\begin{document}

\title{Parameterized Algorithms for Zero Extension and Metric Labelling Problems}

\author[1]{Felix Reidl}
\author[1]{Magnus \Wahlstrom}
\affil[1]{Royal Holloway, University of London, TW20 0EX, UK \\
         \texttt{(Felix.Reidl|Magnus.Wahlstrom)@rhul.ac.uk}}
\date{}
\maketitle

\begin{abstract}
  We consider the problems \Problem{Zero Extension} and \Problem{Metric Labelling}  
  under the paradigm of parameterized complexity. These are natural, well-studied problems
  with important applications, but have previously not received much attention from
  parameterized complexity. 
  
  Depending on the chosen cost function~$\mu$, we find that different
  algorithmic approaches can be applied to design \FPT-algorithms: for arbitrary~$\mu$
  we parameterized by the number of edges that cross the cut (not the cost)
  and show how to solve \Problem{Zero Extension} in time~$O(|D|^{O(k^2)} n^4 \log n)$ using
  randomized contractions. We improve this running time with respect to both parameter and
  input size to~$O(|D|^{O(k)} m)$ in the case where~$\mu$ is a metric.
  We further show that the problem admits a polynomial \emph{sparsifier}, 
  that is, a kernel of size~$O(k^{|D|+1})$ that is \emph{independent of the metric}~$\mu$.

  With the stronger condition that~$\mu$ is described by the distances of
  leaves in a tree, we parameterize by a \emph{gap parameter}~$(q - p)$ between
  the cost of a true solution~$q$ and a `discrete relaxation'~$p$
  and achieve a running time of~$O(|D|^{q-p} |T|m + |T|\phi(n,m))$
  where~$T$ is the size of the tree over which~$\mu$ is defined and~$\phi(n,m)$
  is the running time of a max-flow computation. We achieve a similar running
  for the more general \Problem{Metric Labelling}, while also allowing~$\mu$
  to be the distance metric between an arbitrary subset of nodes in a
  tree using tools from the theory of VCSPs. We expect the methods used
  in the latter result to have further applications.
\end{abstract}

\section{Introduction}\label{sec:intro}
The task of extending a partial labelling of a few data points to a full data
set while minimizing some error function is a natural computational step for
scientific and engineering tasks. For the particular case of data imposed with
a (binary) relationship, we find that the problems \Problem{Zero Extension}
and \Problem{Metric Labelling} are well-suited for optimization in image
processing\footnote{The recent advent of convolutional neural networks seems
to have lessened the importance here, but \Problem{Metric Labelling} is still
used (for a recent example see \eg
\cite{ApplInpainting})}~\cite{DirtyPictures}, Markov Random
Fields~\cite{MetricLabellingApprox}, social network
classification~\cite{ApplSocNetClassifier,ApplHypertextClassifier}, or sentiment
analysis in natural language processing~\cite{ApplRatingInference}.

The problem settings are as follow. For \Problem{Zero Extension}, we are given
a graph~$G$ and a partial labelling~$\tau\colon S \to D$, for some set of
terminals~$S \subseteq V(G)$, alongside a cost function~$\mu\colon D \times D
\to \mathbb R^+$. Our task is to compute a labelling~$\lambda\colon V(G) \to
D$ which agrees with~$\tau$ on~$S$, subject to the following cost: 
for each edge~$uv \in G$ we pay the cost~$\mu(\lambda(u),\lambda(v))$. 
In \Problem{Metric Labelling}, we are given a graph~$G$, a cost function~$\mu$ as
above, and a \emph{labelling cost}~$\sigma\colon V(G) \times D \to \mathbb R^+$.
Again we are asked to compute a labelling~$\lambda$ and in addition to the
above edge-costs we now also pay~$\sigma(v, \lambda(v))$ for each
vertex. Note that this model allows us to emulate terminals, 
by making the cost~$\sigma(v, \lambda(v))$ prohibitive for all but the required label $\lambda(v)$. 
Both problems are generalizations of the \Problem{Multiway Cut} problem
(we simply let~$\mu$ be identically one for all distinct pairs), which has
garnered considerable attention from the $\FPT$ community in the past and
formed a crystallization nucleus for the very fruitful research 
of cut-based problems (see \eg \cite{Marx06cuts,KratschW12FOCS,RandomContraction,IwataWY16,MarxR14MCalg}).

We will find application for most of these tools in our results listed
below, but we wish to particularly highlight the use of tools and relaxations 
from \emph{Valued CSPs} (VCSPs) for designing $\FPT$ algorithms under gap parameters.
VCSPs are a general framework for expressing optimisation problems,
via the specification of a set $\Gamma$ of \emph{cost functions}
(also referred to as a \emph{constraint language}).
Many important problems correspond to VCSP for a specific language $\Gamma$,
including every choice of a specific metric for the problems above.
Thapper and \v{Z}ivn\'y~\cite{ThapperZivnyCSP} characterized the languages $\Gamma$ for
which the resulting VCSP is tractable. \looseness-1

The use of a tractable VCSP as a \emph{discrete relaxation} of an
NP-hard optimisation problem has previously been used for the design of 
surprisingly powerful $\FPT$ algorithms~\cite{IwataWY16} 
(see also related improvements~\cite{IwataYY17,Wahlstrom17bias}).
In this paper, we advance this research in two ways. 
First, previous algorithms of this type have required the relaxation
to have a \emph{persistence} property, 
which allows an optimum to be found by sequentially fixing variables. In this paper,
we relax this condition to a weaker \emph{domain consistency} property.
Second, we use a folklore result from VCSP research to restrict
the behaviour of optimal solutions to a VSCP instance
in order to facilitate the proof that the domain consistency 
property holds for the relevant VCSPs.
See Section~\ref{section:tree-metrics} for details.

\textbf{Related work.}
So far, \Problem{Zero Extension} and \Problem{Metric Labelling} have been researched primarily from the perspective of
efficient and approximation algorithms (for a more complete overview and
hardness results we refer to the paper by Manokaran, Naor, Raghavendra and
Schwartz \cite{ManokaranNRS08STOC}).
Kleinberg and Tardos~\cite{MetricLabellingApprox} introduced \Problem{Metric Labelling} 
and provided a $O(\log |S| \log\log |S|)$ approximation. A result by Fakcharoenphol,
Rao, and Talwar regarding embedding general metrics into tree metrics~\cite{TreeMetricApprox}
improves the ratio of this algorithm to~$O(\log |S|)$ and a lower bound
of~$O((\log |S|)^{\sfrac{1}{2}-\eps})$ was proved by Chuzhoy and Naor~\cite{MetricLabellingHardness}.
Karzanov~\cite{Karzanov980Ext} introduced \Problem{Zero Extension} with the
specific case of~$\mu$ being a graph metric, that is, equal to the distance
metric of some graph~$H$. The central question of his work---for which
graphs~$H$ the problem is tractable---was just recently fully answered by
Hirai~\cite{Hirai13SODA}. Picard and Ratliff much earlier showed that an
equivalent problem is tractable on trees~\cite{ZeroExtOnTrees}.
Fakcharoenphol, Harrelson, Rao, and Talwar showed that the problem can be
approximation to within a factor of~$O(\log |S| / \log\log |S|)$
\cite{ZeroExtApprox}. Karloff, Khot, Mehta, Rabani used the approach by
Chuzhoy and Naor to show that no factor of~$O((\log |S|)^{\sfrac{1}{4}-\eps})$
for any~$\eps > 0$ is possible unless~$\NP \subseteq
\QP$~\cite{ZeroExtHardness}. More recently, Hirai and
Pap~\cite{Hirai14MOR,SPaths} studied the problem from a more structural
angle and we make use of their duality result in the
following.

\textbf{Our results.}
In this paper we study both problems from the perspective of parameterized
complexity. As the choice of metric has a strong effect on the complexity of 
the problem, we give a range of results, ranging from the more generally applicable
to the algorithmically stronger, both in terms of running time and parameterization. 
In the most general setting, when~$\mu$ is a general cost function or a metric, we will 
parametrize not by the \emph{cost} of a solution but by the number of 
\emph{crossing edges} which are precisely the bichromatic edges under a
labelling~$\lambda$. This in particular allows us to include the case of
zero cost pairs under~$\mu$. For general cost functions, we employ the technique
of randomized contractions~\cite{RandomContraction} and prove the following:

\begin{restatable}{theorem}{thmgeneral}
  \Problem{Zero Extension} can be solved in time~$O(|D|^{O(k^2)} n^4 \log n)$ where~$k$
  is a given upper bound on the number of crossing edges in the solution.
\end{restatable}

\noindent
When~$\mu$ is a metric, we are able to give a linear-time
$\FPT$ algorithm, while also improving the dependency on the parameter, 
using important separators~\cite{Marx06cuts}:


\begin{restatable}{theorem}{thmmetric}\label{theorem:pushing-algorithm}
  \Problem{Zero Extension} with metric cost functions can be 
  solved in time~$O( |D|^{O(k)} \cdot m)$ where~$k$ 
  is a given upper bound on the number of crossing edges in the solution.
\end{restatable}

\noindent
For the general metric setting, we also have our most surprising result, 
demonstrating that \Problem{Zero Extension} admits a \emph{sparsifier};
that is, we prove that it admits a polynomial kernel \emph{independent of the metric~$\mu$}. This
result crucially builds on the technique of \emph{representative
sets}~\cite{KratschW12FOCS,Marx09-matroid,Lovasz1977}. The exact formulation
of the result is somewhat technical and we defer it to Section~\ref{section:kernel},
but roughly, we obtain a kernel of size~$O(k^{|S|+1})$, independent of~$\mu$,
where~$k$ is again the number of crossing edges.
This result is a direct, seemingly far-reaching generalization of 
the polynomial kernel for \Problem{$s$-Multiway Cut}~\cite{KratschW12FOCS}.

Next, we consider the case when~$\mu\colon D \to \mathbb Z^+$ is induced by
the distance in a tree~$T$ with~$D \subset V(T)$. Here, relaxing the problem
to allow all labels~$V(T)$ as vertex values defines a tractable
discrete relaxation, in the sense discussed above. 
Using techniques from VCSP, we design a gap-parameter algorithm: 

\begin{restatable}{theorem}{thmtrees}
  Let $I=(G, \tau, \mu, q)$ be an instance of \textsc{Zero Extension}  
  where $\mu$ is an induced tree metric on a set of labels $D$ 
  in a tree $T$, and let $\hat I=(G, \tau, \hat \mu, q)$ be the
  relaxed instance. Let $p = \cost(\hat I)$. Then we can solve $I$
  in time $O(|D|^{q-p} |T||D|nm)$. 
\end{restatable}

\noindent
For the further restriction when~$\mu$ corresponds to the distances 
of the leaves~$D$ of a tree~$T$, we obtain an algorithm with a slightly
better polynomial dependence. Moreover, it uses only elementary operations
like computing cuts and flows:

\begin{restatable}{theorem}{thmleaves}
  Let $I=(G, \tau, \mu, q)$ be an instance of \textsc{Zero Extension}  
  where $\mu$ is a leaf metric on a set of labels $D$ 
  in a tree $T$, and let $\hat I=(G, \tau, \hat \mu, q)$ be the
  relaxed instance. Let $p = \cost(\hat I)$. Then we can solve $I$
  in time $O(|D|^{q-p} |T| m + |T| \phi(n,m))$, where~$\phi$ is the
  time needed to run a max-flow algorithm.
\end{restatable}


\noindent
Finally, we apply the VCSP toolkit to \Problem{Metric Labelling} and obtain
a similar gap algorithm (see Section~\ref{section:tree-metrics} for undefined terms). 

\begin{restatable}{theorem}{thmtreesmetric}
  Let $I=(G, \sigma, \mu, q)$ be an instance of 
  \textsc{Metric Labelling} where $\mu$ is an induced tree metric 
  for a tree $T$ and a set of nodes $D \subseteq V(T)$, 
  and where every unary cost $\sigma(v, \cdot)$ admits an
  interpolation on $T$. Let $\hat I=(G, \hat \sigma, \hat \mu, q)$
  be the relaxed instance, and let $\rho=\cost(\hat I)$.
  Then the instance $I$ can be solved in time
  $O^*(|D|^{q-\rho})$. In particular, this applies for any $\sigma$
  if $D$ is the set of leaves of $T$. 
\end{restatable}


\section{Preliminaries}\label{sec:prelims}
\noindent
For a graph~$G=(V,E)$ we will use~$n_G = |V|$ and~$m_G = |E|$
to denote the number of vertices and edges, respectively.
We write~$\dist_G$ for the distance-metric induced by~$G$, that is,
$\dist_G(u,v)$ is the length of a shortest path between vertices $u,v
\in V(G)$. We denote by~$N_G(v)$ and~$N_G[v]$ the open and closed
neighbourhood of a vertex. For a vertex set~$S \subseteq V(G)$ we write
$\delta(S)$ to denote the set of edges with exactly one endpoint in~$S$.
We omit the subscript~$G$ if clear from the context  all these notations.

For a tree~$T$ we call a sequence of
nodes~$x_1x_2\ldots x_p$ a \emph{monotone sequence} if~$x_1 \leq_P x_2 \leq_P
\ldots \leq_P x_p$ where $P$ is a path in~$T$ and~$\leq_P$ is the linear order
induced by~$P$. Note that~$x_i = x_{i+1}$ is explicitly allowed. For two
nodes~$x,y \in T$ we will denote the unique~$x$-$y$-path in~$T$ by~$T[x,y]$.
For a vertex set~$S$, an \emph{$S$-path packing} is a
collection of edge-disjoint paths~$\cP$ that connect pairs of vertices
in~$S$. We will also consider \emph{half-integral} path packings, here every
edge of the graph is allowed to be used by up to two paths.

Let~$D$ be a set of labels. For a graph~$G$ we
call a function~$\tau \colon S \to D$ for $S \subseteq V(G)$ a \emph{partial
labelling} and a function~$\lambda \colon V(G) \to D$ a \emph{labelling}. The
labelling~$\lambda$ is an \emph{extension} of~$\tau$ if~$\lambda$ and~$\tau$
agree on~$S$, that is, for every vertex~$u\in S$ we have that~$\lambda(u) = \tau(u)$. Given a graph~$G$ and a labelling~$\lambda$
we call an edge~$uv \in E(G)$ \emph{crossing} if~$\lambda(u) \neq \lambda(v)$.
A \emph{$\tau$-path packing} is a
collection~$\cP$ of edge-disjoint paths such that every path~$P \in
\cP$ connects to vertices that receive distinct labels under~$\tau$
(and both are labelled).

\subsection{Cost functions, metrics, and extensions}

\noindent
A \emph{cost function} over~$D$ is a symmetric positive function~$\mu\colon D
\times D \to \mathbb R^+$. We call it \emph{simple} if~$\mu(x,x) = 0$. A cost
function is a  a \emph{metric} if further it obeys the triangle inequality and
it is a \emph{tree metric} if it corresponds to the distance metric of a tree.
We derive an \emph{induced tree metric}
from a tree metric by restricting its domain to a subset $D$ of the
nodes of the underlying tree. A \emph{leaf metric} is an induced tree
metric where $D$ is the set of leaves of the tree. 
Given a cost function~$\mu$, we define the cost of a labelling~$\lambda$ of a
graph~$G$ under a cost function~$\mu$ as
\[
  \cost_\mu(\lambda, G) = \sum_{uv \in G} \mu(\lambda(u), \lambda(v)).
\]
With these definitions in place, we can now define the problem in question:

\begin{problem}{Zero Extension}
  \Input & A graph~$G$ with a partial labelling~$\tau$ over a finite domain~$D$,
           a simple cost function~$\mu$ over~$D$ and an integer~$q$. \\
  \Prob  & Does~$G$ admit an extension~$\lambda$ of~$\tau$ such that~$\cost_\mu(\lambda, G) \leq q$?
\end{problem}

\noindent
For cost functions~$\mu$ that are uniform on all non-diagonal values we
recover (up to some constant scaling of the parameter) the problem
\Problem{Multiway Cut}. Picard and
Ratliff proved that for tree metrics, the problem\footnote{They prove it for a
variant of the \Problem{Facility Location} problem.} is solvable in polynomial
time~\cite{ZeroExtOnTrees}. We will call the special case in which the
distance function is restricted to a leaf metric \Problem{Zero Leaf
Extension}.

\section{Cost functions: Randomized Contractions}\label{sec:general}
\noindent
We will apply the framework by Chitnis \etal~\cite{RandomContraction} to show
that the most general case of \Problem{Zero Extension} is in $\FPT$ when parameterized
by the number of crossing edges. Note that in this setting crossing edges could
incur an arbitrary cost, including zero. However, the stronger parameterization
of only counting the number of crossing edges at non-zero cost makes for 
an intractable problem: With such zero-cost edges, we can express the
problem \Problem{$H$-Retraction} for reflexive graphs $H$,
which asks us to compute a retraction of a graph~$G$ into the fixed graph~$H$.
This problem is already \NP-complete for~$H$ being the reflexive
$4$-cycle~\cite{RetractNPHard} and thus \Problem{Zero Extension}
is \paraNP-complete for parameter $k=0$ if parameterized by the number
of non-zero-cost crossing edges (or indeed if parameterized by the
total cost).



A~\emph{$(\sigma,\kappa)$-good}
separation is a partition~$(L, R)$ of~$V(G)$ such that~$|L|, |R| > \sigma$,
$|E(L,R)| \leq \kappa$, and both~$G[L]$ and~$G[R]$ are connected. There exists
an algorithm that finds a $(\sigma,\kappa)$-good separation
in time $O((\sigma+\kappa)^{O(\min(\sigma,\kappa))} n^3 \log n)$ (Lemma 2.2
in~\cite{RandomContraction}) or concludes that the graph is
\emph{$(\sigma,\kappa)$-connected}, that is, no such separation exists.
The following lemma is a slight reformulation of Lemma~1.1 in~\cite{RandomContraction} 
which in turn is based on splitters as defined by Noar \etal~\cite{Splitters}:\looseness-1    

\begin{lemma}[Edge splitter]\label{lemma:edge-splitter}
  Given a set~$E$ of size~$m$ and integers~$0 \leq a, b \leq m$ one can 
  in time~$O((a+b)^{O(\min\{a,b\})} m \log m)$ construct a set family~$\mathcal F$ over~$E$ 
  of size at most~$O((a+b)^{O(\min\{a,b\})} \log m)$
  with the following property: for any disjoint sets~$A, B \subseteq E$ with~$|A| \leq a$ and~$|B| \leq b$
  there exists a set~$H \in \mathcal F$ with~$A \subseteq H$ and~$B \cap H = \emptyset$.
\end{lemma}

\noindent
Let us first demonstrate how \Problem{Zero Extension} can be solved on
such highly connected instances and then apply the `recursive understanding'
framework to handle graphs with good separations. In the following,
let~$I = (G,\tau,\mu,q)$ be the input instance with~$\mu$ being a cost function
over the domain~$D$.

\begin{lemma}\label{lemma:high-connectivity}
  Let~$G$ be $(\sigma,k)$-connected for some~$\sigma > k$. Then we
  can find an optimal solution in time~$O((|D| + 2\sigma k + k)^{O(k)} (n+m) \log n)$.
\end{lemma}
\begin{proof}
  Let~$\lambda \in \opt(I)$ be an optimal solution and let~$E_\lambda$ be the
  crossing edges. We write~$V(E_\lambda)$ to denote the
  endpoints of these edges. Let~$C_0,C_1,\ldots,C_\ell$ be the connected
  components of~$G - E_\lambda$ with~$C_0$ being the largest one.
  Since the sets~$C_i$ have only one label each under~$\lambda$,
  each one contains at most one terminal.

  Since~$G$ is $(\sigma,k)$-connected, we know that $\ell \leq k$ and that all
  components~$C_1,\ldots,C_\ell$ have size at most~$\sigma$ (\cf Lemma
  3.6~\cite{RandomContraction}). We will assume in the following that~$C_0$
  contains more than~$\sigma$ vertices, otherwise $G$ contains less
  than~$\sigma k$ vertices and we find the set~$E_\lambda$ in time~$O( (\sigma
  k)^{2k} )$ by brute-force.

  Otherwise, we proceed by colouring the edges of~$G$ using an edge splitter 
  (details below). Such a colouring is \emph{successful} if
  \begin{enumerate}
    \item the crossing edges~$E_\lambda$ are red;
    \item each component~$C_i$, $1 \leq i \leq \ell$, contains a blue spanning
          tree; and
    \item each vertex~$u \in C_0 \cap V(E_\lambda)$ is contained in a blue tree
          of size at least~$\sigma + 1$.
  \end{enumerate}
  By fixing a collection of (arbitrary) spanning trees for the
  components~$C_i$, $1 \leq i \leq \ell$ and a collection of trees in~$C_0$
  with~$\sigma+1$ vertices that contain the~$\leq k$ boundary vertices~$C_0 \cap V(E_\lambda)$
  we can see that we need to correctly colour a set~$B \subseteq E(G)$ of at most
  $
    (\sigma-1)\ell + \sigma k \leq 2\sigma k
  $ edges blue while colouring a set~$R \subseteq E(G)$ of at most~$k$ edges red.
  We apply Lemma~\ref{lemma:edge-splitter} with~$a = 2\sigma k$ and
  $b = k$ to construct an edge splitter~$\mathcal F$ of size at most
  $O((2\sigma k + k)^{O(k)} \log m)$ in time~$O((2\sigma k + k)^{O(k)} m \log m)$
  for which we are guaranteed that at least one member~$H \in \mathcal F$
  will contain~$B$ while avoiding~$R$.

  It is left to show that we can easily find a solution in successfully
  coloured graph. Let~$G_B$ be the graph induced by the blue edges. We
  will call a component of~$G_B$ \emph{small} if it contains~$q$ or 
  less vertices and \emph{big} otherwise. Our task is to recover the
  solution-induced components~$C_0, C_1, \ldots, C_\ell$. First notice
  that every~$C_i$, $i \geq 1$ must be a small component in~$G_B$ and
  further that all components reachable from~$C_i$ via red edges must
  either be another component~$C_j$, $j \geq 1$, \emph{or} be a big
  component in $G_B$. Thus,
  we `discover' the sets~$C_1,\ldots,C_\ell$ by the following marking 
  algorithm:
  \begin{enumerate}
    \item Guess which terminal~$x$ lies in the big component~$C_0$
    \item Mark all small components of~$G_B$ that contain terminals
          other than~$x$
    \item Repeat exhaustively: mark all small components of~$G_B$ which have a
          red edge into an already marked component.
  \end{enumerate}
  If our initial guess of~$x$ is correct this procedure will exactly mark the
  components~$C_1,\ldots,C_\ell$.
  Indeed, any small component $C_i$ not marked by this process would
  be a small component containing no terminal and with all edges of $E_\lambda$
  connecting to the big component, in which case there is a solution
  with at most the same cost which merges $C_i$ with $C_0$. 
  The same holds for any collection of small components with 
  no red edge to a marked component. 
  From this we can deduce~$C_0$, the crossing
  edges~$E_\lambda$, and~$\lambda$ itself. In case the colouring step was
  unsuccessful or our guess of~$x$ was wrong, the above procedure will produce
  some set of edges~$E_\lambda$ of non-minimal cost.

  This verification step, given a colouring, is possible in time 
  $O((n+m) |S|)$. Since~$G$ is connected, we can assume that~$|S| \leq k$
  and the total running time to identify~$E_\lambda$ is~$(2\sigma k + k)^{O(k)} (n+m) \log n$.
  Given~$E_\lambda$, the final step is to find an optimal assignment. 
  While the assigment for components~$C_i$ containing terminals is fixed,
  we need to try all possible assigments for the remaining components in 
  time~$O(|D|^{\ell} k) = O(|D|^{k} k)$. Taken both steps together yields
  the claimed running time.
\end{proof}

\noindent
With the well-connected cases handled we can now proceed to solve the general
problem.

\thmgeneral*
\begin{proof}
  We will assume that~$G$ is connected (otherwise we compute an optimal
  solution for every connected component), therefore the number of
  terminals~$|S|$ is bounded by~$k+1$. Let~$\sigma := |D|^k+1$. We run the
  algorithm of Chitnis \etal~\cite{RandomContraction} to find a $(\sigma,
  k)$-good separation. Assume for now that such a good separation~$(V_1,V_2)$
  is found with at most~$k$ edges~$E(V_1,V_2)$ crossing it.

  Let~$S_1 := E(V_1,V_2) \cap V_1$ be the~$\leq k$ vertices in~$S_1$ on the
  border of the separation. We iterate through all~$|D|^{|S_1|} \leq |D|^k$
  ways the vertices~$S_1$ could receive colours by a solution. For each such
  colouring~$\tau$, we construct a sub-instance~$G_{\tau}$ by identifying all vertices that have
  the same terminal-labelling and recursively solve the instance. For each
  solution of~$G_{\tau}$ with at most~$k$ crossing edges we collect said
  crossing edges in~$E^*$. That is, the set~$E^*$ contains all edges that are
  crossing for an optimal solution of \emph{some} labelling of~$S_1$.
  Note that~$|E^*| < \sigma$ by our choice of~$\sigma$. Since~$G_1$ is connected,
  there is at least one edge in~$\hat E = E(G_1)\setminus E^*$. Now, every 
  solution $\lambda$ of~$G$ that has at most~$k$ crossing edge can be modified to have
  only edges of~$E^*$ and thus no edges of~$\hat E$ crossing it; we simply
  observe what labelling~$\tau$ the solution~$\lambda$ applies to~$S_1$ and
  replace the colour~$\lambda$ applies to~$V_1$ for the colours that the optimal
  solution of~$G_{\tau}$ applies to it. 

  Consequently, we can assume that the optimal solutions we are interested in
  are not crossed by~$\hat E$, therefore it is safe to contract~$\hat E$ and
  therefore reduce the size of~$V_1$ to~$\sigma$ (recall that~$G[V_1]$ is connected
  and the contraction of course preserves that property). We repeat this procedure
  until the resulting graph is $(\sigma,k)$-connected (alternatively, any
  of the recursive calls could tell us that no solution with at most~$k$ crossing
  edges exists in which case we return that the instance has no solution).
  Then by Lemma~\ref{lemma:high-connectivity}, we can decide the problem in
  time~$O((|D| + 2\sigma k  + k)^{O(k)}(n+m)\log n) = O(|D|^{O(k^2)} n^2 \log n)$. 

  Let us now analyse the total running time~$T(n)$; note that~$k$ and~$\sigma$
  do not change with each recursive call. First, finding a good separation
  takes time $O((\sigma+k)^{O(k)} n^3 \log n) = O(|D|^{O(k^2)} n^3 \log n)$ and constructing the
  instance~$G_{\tau}$ at total of~$O(|D|^{k}(n+m))$, which is dominated by the
  former running time. The recursive call on~$V_1$ with~$n_1 := |V_1|$ costs
  us~$T(n_1)$, after which we are left to work on an instance of size at
  most~$n - n_1 + \sigma$ which will cost us~$T(n - n_1 + \sigma)$. Note that,
  by the properties of~$(\sigma,k)$-good separations, it holds that~$\sigma+1
  \leq n_1 \leq n - \sigma - 1$. We therefore need to resolve the recurrence
  \[
    T(n) \leq \max_{\sigma + 1 \leq n_1 \leq n - \sigma - 1}
    \Big(  
        |D|^{O(k^2)} n^3 \log n + T(n_1) + T(n-n_1+\sigma)
    \Big).
  \]
  As noted in~\cite{RandomContraction}, the maximum in this expression
  is attained at the extreme values for~$n_1$ and that the claimed
  running time is a bound on~$T(n)$.
\end{proof}

\section{General metrics: Pushing separators}

We now consider the more restricted, but reasonable case that $\mu$
is a metric, observing the triangle inequality. We find that this
allows a `greedy' operation of \emph{pushing} in a solution $\lambda$,
which allows both the design of a faster algorithm (Section~\ref{section:pushing-algorithm})
and the computation of a \emph{metric sparsifier} (Section~\ref{section:kernel}).

\subsection{The pushing lemma}

Let in the following $I=(G=(V,E), \tau, \mu, q)$ be an instance of
\textsc{Zero Extension} for an arbitrary metric $\mu$. Let $S \subseteq V$ be
the range of $\tau$, and let $D$ be the set of labels. We assume that the
following reductions have been performed on $G$: For every label $\ell$ used
by $\tau$ there is a terminal $t_\ell$, and every vertex $v$ such that
$\tau(v)=\ell$ has been identified with this terminal $t_\ell$.


Let $\lambda\colon V \to D$ be an extension of $\tau$, and let 
$U = \lambda^{-1}(\ell)$ for some $\ell \in D$. 
By \emph{pushing from $\ell$ in $\lambda$} we refer to the operation
of relabelling vertices to grow the set $U$ ``as large as possible'',
without increasing the number of crossing edges.
Formally, this refers to the following operation: Let $C$ be the
furthest min-cut between vertex sets $U$ and $S - t_\ell$
(respectively $S$ if there is no terminal $t_\ell$),  
let $U'$ be the vertices reachable from $U$ in $G-C$, and 
let $\lambda'$ be the labelling where $\lambda'(v)=\ell$
for $v \in U'$ and $\lambda'(v)=\lambda(v)$ otherwise.
Clearly, $\lambda'$ is an extension of $\tau$. The purpose of this
section is to show that as long as $\mu$ is a metric (in particular,
observes the triangle inequality), this operation does not increase
the cost of the solution.

\begin{lemma}[Pushing Lemma] \label{lemma:all-pushing-allowed}
  For any $\tau$-extension $\lambda$ and every label $\ell \in D$,
  pushing from $\ell$ in $\lambda$ yields a $\tau$-extension $\lambda'$ 
  with $\cost_\mu(\lambda',G) \leq \cost_\mu(\lambda,G)$. 
\end{lemma}
\begin{proof}
  Write $V_\ell=\lambda^{-1}(\ell)$, and for a set of 
  vertices $U$ let $\delta(U)$ denote the edges with one endpoint in
  $U$.  Let $\cP$ be a max-flow from $V_\ell$ to the set 
  $S-t_\ell$, let $C$ be a furthest corresponding min-cut,
  and let $V_\ell^+ \supseteq V_\ell$ be the set of vertices reachable
  from $t_\ell$ in $G'-C$. By Menger's theorem, $\cP$ partitions 
  the set of edges $\delta(V_\ell^+)$. Finally, let $\cP^-$ consist 
  of the prefixes of the paths $P \in \cP$ up until and including 
  the edges of $C$. For $P \in \cP^-$, let $\lambda(P)$ be 
  the label of its final edge. Let $\lambda'$ be the assignment
  resulting from letting $\lambda'(v)=\ell$ for every $v \in V_\ell^+$
  and $\lambda'(v)=\lambda(v)$ otherwise. We make two quick
  observations. 

  \begin{claim}
    The cost incurred by $\lambda'$ on every path $P \in \cP^-$ 
    is precisely $\mu(\ell, \lambda(P))$, whereas for every edge $uv$
    not on any such path the cost is at most as high as for $\lambda$,
    i.e., $\mu(\lambda'(u),\lambda'(v)) \leq \mu(\lambda(u),\lambda(v))$.
  \end{claim}
  \begin{proof}
    For every path $P \in \cP^-$ only the final vertex $v$
    has a label $\lambda'(v) \neq \ell$, and the final edge has cost 
    $\mu(\ell, \lambda(v))$ where $\lambda(v)=\lambda(P)$. 
    For the second part, only edges with at least one end point in
    $V_\ell^+$ have changed cost from $\lambda$ to $\lambda'$, 
    and among such edges only the edges of $C$ have non-zero cost. 
    Since $C$ is a min-cut, all such edges are covered by paths 
    $P \in \cP^-$. 
  \end{proof}
  
  \begin{claim}
    The cost incurred by $\lambda$ on a path $P \in \cP^-$
    is at least $\mu(\ell, \lambda(P))$. 
  \end{claim}
  \begin{proof}
    Let $P=t_\ell v_1 \ldots v_r$ where $v_r$ may be a terminal. 
    This describes a walk $\lambda(t_\ell)=\ell$, \ldots,
    $\lambda(v_r)=\lambda(P)$ from $\ell$ to $\lambda(P)$,
    and every edge $uv$ of $P$ where $\lambda(u)\neq \lambda(v)$ 
    incurs the corresponding cost. By the triangle inequality for $\mu$, 
    the sum of these costs is at least $\mu(\ell, \lambda(P))$. 
  \end{proof}

  \noindent
  The result follows immediately from the above two claims.
\end{proof}

\noindent
An immediate consequence of the above lemma is the following
reduction rule.

\begin{corollary} \label{corollary:pushing-reduction}
  We can reduce to the case where for 
  every label $\ell$ used by $\tau$, there is a single terminal
  $t_\ell$ of value $\tau(t_\ell)=\ell$ such that $\delta(t_\ell)$
  is the unique isolating min-cut for $t_\ell$. 
\end{corollary}
\begin{proof}
  The first part of the reduction (to terminals $t$) is trivial.
  Let $S$ be the set of terminals. 
  Let $t_\ell$ be a terminal with label $\ell$, let $C_0$ be a
  furthest isolating min-cut between $t_\ell$ and $S-t_\ell$,
  and let $V_\ell^0$ be the vertices reachable from~$t_\ell$ in~$G-C_0$. 
  
  Now let $\lambda\colon V \to D$ be a $\tau$-extension and 
  let $V_\ell = \lambda^{-1}(\ell)$. Let $V_\ell^+ \supseteq V_\ell$
  be produced by pushing from $\ell$ in $\lambda$.
  By Lemma~\ref{lemma:all-pushing-allowed}, replacing $V_i$
  by $V_i^+$ does not incur a larger cost. 

  To see that our reduction is valid, let $f\colon 2^V \to \Z$ be the edge-cut
  function for $G$. Then, by submodularity,
  \[
    f(V_\ell^0) + f(V_\ell^+) \geq 
    f(V_\ell^0 \cap V_\ell^+)  + f(V_\ell^0 \cup V_\ell^+) \geq
    f(V_\ell^0) + f(V_\ell^0 \cup V_\ell^+),
  \]
  where the last step follows since $\delta(V_\ell^0)$ is a 
  minimum isolating cut. Thus $V_\ell^0 \cup V_\ell^+$ is
  a cut between $t_\ell$ and $S-t_\ell$ of cost no more
  than $V_\ell^+$, which implies that they are equal and
  $V_\ell^0 \subseteq V_\ell^+$. Hence any solution $\lambda$
  can be modified to let $\lambda(V_\ell^0)=\ell$, 
  and we may contract all vertices $V_\ell^0$ into $t_\ell$. 
  Repeating this for all vertices yields an instance as described. 
\end{proof}

\subsection{An \textsf{FPT} algorithm}
\label{section:pushing-algorithm}
Let $I=(G, \tau, \mu, q)$ be an instance of \textsc{Zero Extension} where
$\mu$ is a metric, \eg a simple cost function observing the triangle
inequality.  Let $D$ be the domain of $\mu$, and let $S$ be the terminals in
$G$. We will show that the problem is $\FPT$ parameterized by $k+|D|$ where
$k$ is a bound on the number of crossing edges in an optimum $\lambda$.

The algorithm uses the pushing lemma (Lemma~\ref{lemma:all-pushing-allowed})
to guess a solution $\lambda$ using the technique of \emph{important separators}. 
This is a classical ingredient for $\FPT$ algorithms for cut problems,
pioneered by Marx~\cite{Marx06cuts}. We focus on the edge version.
Let $G=(V,E)$ be a graph with disjoint vertex sets $S$ and $T$,
and let $C \subseteq E$ be a minimal $(S,T)$-cut which is not
necessarily minimum. Let $U$ be the set of vertices reachable from $S$
in $G-C$. Then $C$ is an \emph{important separator} if for every set
$U' \supset U$ with $T \cap U' = \emptyset$ we have
$|\delta(U')|>|C|$.  In other words, $C$ represents a greedy
``furthest cut'' from $S$ for its size. The important realization
is that for every bound $k$ on $|C|$, there are at most $f(k)=4^k$
important $(S,T)$-separators in $G$. 

The algorithm we will use is inspired by the classical algorithm for
\textsc{Multiway Cut} of Chen, Liu and Lu~\cite{ChenLL09MWC},
improving on Marx~\cite{Marx06cuts}. 
This algorithm works by repeatedly selecting a
terminal $t$, guessing an important separator around $t$ (against the
other terminals $S-t$), then deleting the chosen separator and
proceeding with the next terminal with a decreased budget $k$. 
The important aspect to us is the process of enumerating important 
separators, which for edge-cuts can be described as follows.
\begin{enumerate}
\item Assume that we are enumerating important separators between 
  sets $S$ and $T$, which may be singleton sets (the process is the
  same regardless). Assume that $k$ is our budget for the maximum
  separator size. 
\item Compute a furthest min-cut $C$ between $S$ and $T$. If $|C|>k$, 
  abort. Otherwise, contract edges so that $\delta(S)=C$. Initialise
  all edges as unmarked. 
\item Select an unmarked edge $uv \in \delta(S)$ and branch
  recursively on two cases:
  \begin{enumerate}
  \item Contract $uv$ into $S$, thereby increasing the max-flow
  \item Mark $uv$ as part of the final separator
  \end{enumerate}
  Abort a branch whenever the resulting max-flow is more than $k$. 
\item Once all edges of $\delta(S)$ are marked, output $\delta(S)$  as
  an important separator. 
\end{enumerate}
For the running time, we may analyse this in terms of the gap between
$|C|$ and $k$. More precisely, consider an alternative lower bound
where every marked edge of $C$ counts for 1 point, but unmarked edges
of $C$ are worth 1/2 point. Then it is clear that the lower bound
increase by $1/2$ in both branches of the recursive calls, and
at most $2^{2k}$ important separators are generated. 


Compared to \textsc{Edge Multiway Cut}, there are two complications to
an algorithm for \textsc{Zero Extension}. First, if $C$ are the
crossing edges on an optimal solution $\lambda$, it may be that not
every connected component of $G-C$ contains a terminal;
therefore the algorithm may need to branch on non-terminal vertices as
well. This is not an obstacle in itself, but it allows for only a
weaker running time analysis. Second, even after $C$ has
been identified, it remains to find an assignment $\lambda$ that
minimizes cost. The complexity of this may vary depending on the
particular metric $\mu$. 

For this reason, we describe the algorithm below as being composed of
stages, where the first stage identifies all crossing edges reachable
from a terminal, the second stage identifies the remaining crossing
edges, and the third stage finds an assignment $\lambda$. For specific
metrics $\mu$, it may be possible to skip or speed up the second and
third stages; e.g., for a leaf metric it can be checked that
the algorithm is finished after the first stage. 

In summary, we show the following. 

\thmmetric*

\noindent
We begin by proving the running time for the first stage. 
This is analysed in terms of a lower bound $p$ on the crossing number
of any labelling $\lambda$. This bound is computed as follows.
Let $S$ be the set of terminals. Observe that for any
$\tau$-extension $\lambda$, the crossing edges in $\lambda$ form a
multiway cut for $S$ in $G$. This can be lower-bounded as follows.
For every $t \in S$, let $f_t$ be the value of a max-flow between $t$
and $S-t$, and let $p=\sum_{t \in S} f_t/2$. 
Then there is a half-integral packing of $p$ terminal-terminal paths
in $G$ (i.e., a $\tau$-path-packing for the tree which is a star with
$S$ as leaves, as in
Section~\ref{section:leaf-metric})~\cite{SchrijverBook}, hence there
is no solution $\lambda$ with crossing number less than $p$. 
Then the first phase can be computed in time parameterized by the gap 
$k-p$. (We note that this is a weaker lower bound, and hence a weaker
gap parameter, compared to the relaxation lower bound $\rho$ used
later in this paper. However, it of course applies to any metric.)

\begin{lemma}
  Let $p$ be the lower bound as above 
  In $O(4^{k-p}km)$ time and $4^{k-p}$ guesses, we can reduce to
  the case where every edge of $\delta(t)$ is a crossing edge in the
  optimal solution for every $t \in S$. 
\end{lemma}
\begin{proof}
  This phase works almost exactly as \textsc{Edge Multiway Cut}. 
  By Corollary~\ref{corollary:pushing-reduction}, we assume that 
  for every $t \in S$, $\delta(t)$ is the unique ($t, S-t)$-min cut. 
  Since our parameter is a cardinality parameter, we implement this
  step by a standard augmenting path algorithm, aborting whenever more
  than $k$ paths have been found. More specifically, we proceed as
  follows.  Let $f$ trace the number of paths found in total across
  all terminals; initially $f=0$. Then for each terminal $t \in S$ in
  turn, we pack $t-(S-t)$-paths in $G$, increasing $f$ at each point,
  and aborting if $f>2k$ is reached. Assuming this does not happen, 
  we can in $O(m)$ time compute the furthest $(t, S-t)$-min cut $C$,
  and the set $V_t$ of vertices reachable from $t$ in $G-C$, and
  finally the new graph $G$ resulting from contracting $V_t$ into
  $t$. In total, this process takes $O(km)$ time and either aborts
  or produces a reduced graph. We assume in the rest of the proof that
  the process has succeeded and proceed as follows.
  
  Initialize all edges $tt'$ for $t, t' \in S$ as marked, all other
  edges as unmarked.  Compute a lower bound $p$ on $k$ by counting 
  1 for every marked edge, and $1/2$ for every unmarked edge of
  $\delta(S)$. If $p>k$, reject the instance. Otherwise, proceed 
  as in the enumeration of important separators, \ie select an
  unmarked edge $tv$ incident with some $t \in S$ and branch on two
  cases: Either contract $tv$ into $t$, and recompute the max-flow and
  value of $p$; or mark $tv$ as part of the solution and select another
  edge. Whenever $p>k$, abort the branch; whenever  every edge of
  $\delta(S)$ is marked, return the instance. 

  For the correctness, the value $p$ is a lower bound on the crossing
  number of $\lambda$, as argued above. Hence if $p>k$, the current
  instance has no solution. Otherwise, the branching is
  exhaustive, marking $tv$ as either crossing or non-crossing, and
  in the latter case we are allowed to contract out to the new furthest
  min-cut by the pushing lemma (Lemma~\ref{lemma:all-pushing-allowed}). 
  Hence, assuming the instance has a
  solution at all, in at least one output instance the marked edges
  correspond to the crossing edges of some optimal labelling
  that are reachable from a terminal.

  Regarding running time, Corollary~\ref{corollary:pushing-reduction}
  can be applied in $O(|S| k m)$ time by an augmenting path
  approach. For every further branching step, recomputing a new
  max-flow can be done in $O(k m)$ time. Throughout the branching
  process, the value of $k-p$ decreases by at least $1/2$ in every
  branch (in particular, the max-flow number increases in the 
  contraction branch). Hence the branching process produces at most
  $2^{2(k-p)}$ outputs.
\end{proof}

\noindent
Next, we show a similar branching algorithm (without a lower bound)
to mark the remaining crossing edges of a solution. 

\begin{lemma}
  Given an input from stage 1, with $p$ edges already marked, 
  in $O(4^{2k-p}m)$ time and $4^{2k-p}$ guesses we can reduce
  to the case where every edge of $G$ is crossing in the optimal
  solution.  
\end{lemma}
\begin{proof}
  We assume that the input from stage 1 has the property that an edge
  is marked if and only if it is incident with a terminal. We proceed
  with a branching process as follows. Let $v$ be an arbitrary
  non-terminal vertex incident with at least one non-marked edge;
  if none exists, simply output the instance. 
  Compute a furthest min-cut $C$ between $v$ and $S$ in $G$, aborting
  if $|C|>k$. Let $U$ be the vertices reachable from $v$ in
  $G-C$. If $U$ contains any marked edge, abort the branch as being
  the result of inconsistent choices; otherwise contract $U$ into $v$.
  If every edge of $\delta(v)$ is marked, proceed with a different
  starting vertex $v$. Otherwise, as above select one unmarked 
  edge from $\delta(v)$ and branch on either marking it or contracting
  it and recomputing the max-flow and min-cut. If at any point more
  than $k$ edges have become marked, abort the branch. Once only
  marked edges remain, output the instance.

  The correctness argument is similar as in stage 1. Assume that we
  are currently working with a non-terminal vertex $v$, and that
  $\lambda(v)=i$ in an optimal solution to the current instance;
  by Lemma~\ref{lemma:all-pushing-allowed}, we may assume that pushing
  from $i$ in $\lambda$ has no effect. Let $V_i=\lambda^{-1}(i)$
  and let $V_v$ be the vertices of the connected component of $G[V_i]$
  that contains $v$. Then $\delta(V_i)$ is an important separator, and
  it follows that the same holds for $\delta(V_v)$ (as otherwise
  pushing would produce a bigger set $V_i$). Thus contracting it to a
  furthest min-cut is allowed. After this, the branching process is
  exhaustive as above. 

  For the running time, we view edges as being \emph{double-marked},
  receiving one mark from each endpoint. Edges marked in stage 1 are
  viewed as having received a mark from each terminal side, \ie at
  least $p$ marks have already been placed at the start of the
  algorithm. We then view the branching around a vertex $v$ as placing
  one such mark on the $v$-side of the solution edges connecting to
  $V_i$ in $V_v$.  Hence, both the packing of a path from $v$ and the
  final marking of an edge represents the placing of half a mark, and
  in total at most $2k-p$ such marks will be placed (ignoring the
  final leaf of an aborted branch where $|C|$ grows too large). Hence
  the total number of branches in this stage is at most $2^{2(2k-p)}$. 
  Finally, for the polynomial part of the running time, we consider
  the amount of work done in a single node, before branching. As
  noted, every additional path added in the max-flow computation takes
  $O(m)$ time to find and counts as a half-mark against our budget,
  hence at most $O(k)$ paths are found before aborting. Additionally, 
  every contraction step (out to a furthest min-cut) takes $O(m)$ time
  in total, and by the same argument only $O(k)$ contraction steps
  are performed. Hence the local work per node is bounded as
  $O(km)$.
\end{proof}

\noindent
After stage 2, the remaining graph contains at most $k$ edges, 
hence at most $O(k)$ vertices, and it only remains to find the
min-cost labelling of the non-terminal vertices. In the absence of 
any stronger structural properties of the metric $\mu$, this last
phase can be completed in $|D|^{O(k)}O(m)$ time. 
Theorem~\ref{theorem:pushing-algorithm} follows.

%


\subsection{A kernel for any metric}\label{section:kernel}
We show that \textsc{Zero Extension} has a kernel of $O(k^{s+1})$
vertices for any metric $\mu$, where $k$ is a bound on the crossing
number of a solution and $s$ is the number of labels of $\mu$. In
fact, stronger than this, we show that such a kernel can be computed
without access to $\mu$: We can find a set $Z$ of $O(k^{s+1})$ edges
such that for every instance $I=(G, \tau, \mu, q)$ with terminal set
$S$, if $I$ admits any solution with at most $k$ crossing edges, 
then $F$ contains all crossing edges of at least one optimal solution
of this type for $I$. By contracting all edges not in $Z$ this allows 
us to construct a graph $G'$ with $O(k^{s+1})$ edges such that $(G',S)$
has the `same behaviour' as $(G,S)$, up to the values of $k$ and $s$. 
We refer to this as a \emph{metric sparsifier}.
Let us make this more precise.

\begin{definition}
  For an instance $I=(G, \tau, \mu, q)$ of \textsc{Zero Extension}
  and an integer $k$, the \emph{$k$-bounded cost} $\cost(I,k)$ of $I$
  is the minimum cost of a $\tau$-extension $\lambda$ with crossing
  number at most $k$, or $\infty$ if no such $\tau$-extension exists.  
  Let $G=(V,E)$ be a graph with a set of terminals $S \subseteq V$,
  and let $k$ and $s$ be integers, $s \geq |S|$. A \emph{$k$-bounded
    metric sparsifier for $(G, S)$ (for metrics with up to $s$ labels)}
  is a graph $G'=(V',E')$ with $S \subseteq V'$, such that for
  any metric $\mu$ on a set $D$ of at most $s$ labels, and for any
  injective labelling $\tau\colon S \to D$, we have
   $\cost((G, \tau, \mu, q),k) = \cost((G', \tau, \mu, q), k)$. 
\end{definition}

\noindent
We show the following result, which implies the existence of a small
metric sparsifier. 

\begin{theorem} \label{theorem:kernel}
  Let $s \geq 3$ be a constant. For every graph $G=(V,E)$ 
  with a set $S$ of terminals, $|S| \leq s$, and integer $k$, 
  there is a randomized polynomial-time computable set $Z \subseteq E$
  with $|Z|=O(k^{s+1})$ such that for any instance
  $I=(G, \tau, \mu, q)$ of \textsc{Zero Extension}
  with $S$ being the set of terminals in $I$ and $\mu$ having 
  at most $s$ labels, if $\cost(I, k)<\infty$ then
  there is a $\tau$-extension $\lambda$ with crossing number at most
  $k$ and cost $\cost(I, k)$ such that every crossing edge of
  $\lambda$ is contained in $Z$. 
\end{theorem}

\noindent
The kernel is an adaptation of the \emph{irrelevant vertex} strategy 
used in the kernel of \textsc{Multiway Cut} for $s$ terminals of
Kratsch and  Wahlstr\"om~\cite{KratschW12FOCS}. In that paper, a
kernel is produced by first computing a set $Z_0$ of $O(k^{s+1})$
vertices which contains all vertices $v$ which are contained in \emph{every} minimum
solution of size at most $k$. Then one \emph{irrelevant vertex} 
$v \notin Z_0$ is chosen and removed from the graph via a
bypassing operation, and the process is repeated until the computed
set $Z_0$ covers all non-terminal vertices, at which point $Z_0$
is the desired final set $Z$. The set $Z_0$ is
computed using a tool from matroid-theory called
\emph{representative sets}. These same components will yield a
kernel for our present problem. 

We give only a brief sketch of the technical background, and focus on
the kernelization result itself. 
For a complete description of the technical tools involved,
see Kratsch and Wahlstr\"om~\cite{KratschW12FOCS}.

A \emph{matroid} is an ``independence system'' $M=(E, \cI)$,
$\cI \subseteq 2^V$, subject to certain axioms.
The sets $S \in \cI$ are the \emph{independent sets} of $M$.
Matroids have broad applications in combinatorics in general; see
Oxley~\cite{OxleyBook} and Schrijver~\cite{SchrijverBook}. 
A \emph{representation} of a matroid is a matrix $A$ with columns labelled by $E$
such that for every $S \subseteq V$, $S \in \cI$ if and only if
the corresponding columns of $A$ are linearly independent.
A \emph{gammoid} is a matroid corresponding to flows in a graph. 
Given a directed graph $G=(V,E)$ with a set of source vertices 
$S \subseteq V$, and a set $U \subseteq V$, the gammoid defined from
$G$, $S$ and $U$ is a matroid $M=(U, \cI)$ where $T \subseteq U$ is
independent if and only if the fully vertex-disjoint max-flow from
$S$ to $T$ is of size $|T|$; equivalently, there is no $(S,T)$-cut
in $G$ of fewer than $|T|$ vertices. A representation of a gammoid
can be computed in randomized polynomial time~\cite{OxleyBook}.
In our application, we need the gammoid to represent edge cuts
instead of vertex cuts; clearly this can be done by subdividing
edges of $E$ and multiplying every vertex $v$ of $G$ into $d(v)$
copies. Undirected edges of $G$ can be implemented as a pair of
directed edges in opposite directions. Let $G'$ be the directed
graph resulting from applying these modifications to $G$. 
  
The main tool in our kernel is the following result. (See also Fomin
et al.~\cite{FominLPS16} for a faster algorithm.)

\begin{lemma}[\cite{Lovasz1977,Marx09-matroid}]%
  \label{lemma:representative-sets}
  Let $M=(E, \cI)$ be a linear matroid represented by a matrix $A$ of rank $r+s$,
  and let $\cY$ be a collection of independent sets of $M$, each of size $s$. 
  Assume that $s$ is a constant. 
  Then in polynomial time we can compute a set $\cY^*\subseteq \cY$ 
  of size at most $\binom{r+s}{s}$ such that for every
  set $X \subseteq E$, there is a set $Y \in \cY$
  such that $X \cap Y=\emptyset$ and $X \cup Y \in \cI$
  if and only if there is such a set $Y' \in \cY^*$. 
\end{lemma}

\noindent
As shorthand, for an independent set $X$ we say that $Y$
\emph{extends} $X$ if $X \cap Y=\emptyset$ and $X \cup Y \in \cI$. 
The following is a useful characterization of this notion in gammoids.

\begin{proposition}[Prop.~1 of~\cite{KratschW12FOCS}]
  \label{prop:extends}
  Let $X$ be a independent set in a gammoid defined from $G$, $S$ and
  $U$. Let $X'$ be the minimum $(S,X)$-vertex cut closest to $S$ in $G$,
  which may overlap $S$. Then for a vertex $v \in U$, the set $\{v\}$
  extends $X$ in the gammoid if and only if $v$ is reachable from $S$
  in $G-X'$.
\end{proposition}

\noindent
For further terms from matroid theory used below,
see~\cite{KratschW12FOCS}, alternatively Oxley~\cite{OxleyBook}
and Marx~\cite{Marx09-matroid}.

We are now ready to prove Theorem~\ref{theorem:kernel}.

\begin{proof}
  Let a graph $G=(V,E)$ with a terminal set $S \subseteq V$,
  $|S| \leq s$, and an integer $k$ be given. We apply
  Corollary~\ref{corollary:pushing-reduction} to $G$, 
  then create a digraph $G'$ from $G$ as above. Additionally, for every edge 
  $e=uv \in E$ we introduce a \emph{sink-copy} $e'$ of $e$ in
  $G'$, with in-arcs from all copies of vertices $u$ and $v$ in $G'$,
  and no further in- or out-arcs.  
  Let $E_S=\bigcup_{t \in S} \delta(t)$. We first note that if
  $|E_S|>2k$ then we may reject the input.

  \begin{claim}
    If $\sum_{t \in S} d(t,G) > 2k$, then every labelling $\lambda$
    that is injective on $S$ has more than $k$ crossing edges.
  \end{claim}
  \begin{proof}
    The crossing edges of any such $\lambda$ form a multiway cut of
    $(G,S)$. It is known that the cardinality of a multiway cut
    is at least $\sum_{t \in S} d(t,G)/2$, as noted previously in
    Section~\ref{section:pushing-algorithm}.
  \end{proof}

  Hence if $|E_S|>2k$, then we reject the input (alternatively, simply
  produce $Z=\emptyset$). Otherwise proceed as follows. Let an
  \emph{instance over $(G,S)$} be an instance 
  $I(\tau, \mu, q)=(G, \tau, \mu, q)$ of \textsc{Zero Extension} 
  where $\mu$ is a metric over a set of labels $D$ with $|D| \leq s$
  and $\tau$ is an injective labelling $\tau\colon S \to D$. 
  We first observe that the existence of a $\tau$-extension $\lambda$
  with crossing number at most $k$ for any instance over $(G,S)$ is a
  property purely of $(G, S)$ and $k$, hence independent of the
  metric; explicitly, it exists if and only if $(G,S)$ has a multiway
  cut of at most $k$ edges. Since the theorem is vacuous otherwise, 
  we assume that such a multiway cut exists, hence that
  $\cost(I, k) < \infty$ for every instance $I$ over $(G,S)$. 
  Say that an edge $e$ is \emph{essential} in $G$ if there is
  some instance $I=I(\tau, \mu, q)$ over $(G,S)$ such that there is at
  least one $\tau$-extension $\lambda$ with crossing number at most
  $k$ and cost at most $q$, and the edge $e$ is crossing in every such
  $\lambda$. We compute a set $Z_0$ that contains all essential edges; 
  any edge of $G$ not contained in $Z_0 \cup E_S$ is then
  \emph{irrelevant}. The computation of $Z_0$ makes up the major part
  of this proof. 

  Number the vertices of $S$ as $S=\{t_1, \ldots, t_r\}$, $r=|S|$. 
  We define a matroid $M$ as the disjoint union of the following
  matroids: Matroids $M_1$ through $M_s$ are disjoint copies of the
  gammoid defined from $G'$ with source set 
  $E_S=\bigcup_{t \in S} \delta(t)$.  Note that $E_S$ is a vertex set
  in $G'$. Finally, let $M_0$ be the uniform matroid $U_{m,k}$ of rank
  $k$ on ground set $E(G)$. Since $M_0$ has a representation over
  every sufficiently large field, we can compute a representation of
  $M$ in randomized polynomial time with exponentially small failure
  probability~\cite{Marx09-matroid}. We refer to $M_1$, \ldots, $M_s$
  and $M_0$ as the \emph{layers} of $M$.
  For an independent set $X$ of $M$,
  a set $Y$ extends $X$ in $M$ if and only if the restriction of $Y$  
  to layer $i$ extends the restriction of $Y$ to layer $i$
  for every layer $i$. 
  Since the rank of each gammoid is $|E_S| \leq 2k$, the rank
  of $M$ is at most $(2s+1)k=O(k)$ since $s$ is a constant.
  For an edge $e \in E(G)$, we let $e_i$ (respectively $e_i'$) refer to
  the copy of $e$  (of $e'$) in $M_i$, $i \in [s]$, and $e_0$ refers to
  the copy of $e$ in $M_0$. 

  We now define the collection $\cY$ of sets of size $s+1$. 
  For an edge $e \in E(G)$, let $Y(e)=(e_1', \ldots e_s', e_0)$, and 
  define $\cY=\{Y(e) \mid e \in E(G) \setminus E_S\}$. 
  Compute a representative set $\cY^* \subseteq \cY$ in $M$,
  and let $Z_0 \subseteq E = \{e \in E \mid Y(e) \in \cY^*\}$. 
  Then $|Z_0|=O(k^{s+1})$ by Lemma~\ref{lemma:representative-sets}.
  We show that $Z_0 \cup E_S$ contains every edge that is essential in
  $G$.  

  \begin{claim}
    If $e$ is essential in $G$ for some instance $I=I(\tau, \mu, q)$
    over $(G,S)$, then $e \in Z_0 \cup E_S$. 
  \end{claim}
  \begin{proof}
    Let $D$ be the set of labels of $\mu$.
    Let $\lambda$ be an extension of $\tau$ with at most $k$ crossing
    edges, and of cost $\cost(I, k)$. Assume that among all such
    extensions, $\lambda$ has the minimum crossing number.
    Let $C$ be the set of crossing edges in $\lambda$,
    and for $i \in D$ let $V_i=\lambda^{-1}(i)$. 
    Hence $C=\bigcup_{i \in D} \delta(V_i)$. 
    Define a set $X$ in $M$ as follows: 
    In layers $i=1, \ldots, r$, $X$ contains the copy
    of edges $\delta(V_i) \cup \delta(t_i)$; 
    in layers $i=r+1, \ldots, s$, $X$ contains the copy
    of edges $\delta(V_i)$; and in the final layer
    (representing $M_0$), $X$ contains the copy of $C-e$. 
    It is clear that $X$ is independent in $M$, as otherwise
    there is some label $i$ such that (by
    Lemma~\ref{lemma:all-pushing-allowed}, the pushing lemma) pushing 
    from some label $i$ would yield a labelling $\lambda'$ with
    smaller crossing number and at most the same cost.      
    The claim is now that $Y(f)$ extends $X$ if and only if $f=e$. 

    In the one direction, assume that $Y(f)$ extends $X$. 
    If $f \in C$, then $f=e$ by $M_0$. Otherwise, $f$ lies 
    in some set $V_i$, $i \in D$. But then the restriction of $X$ 
    to layer $i$ separates $f$ from $E_S$, so $f'$ cannot possibly
    extend $X$ in this layer. Hence $f \in C$, and $f=e$. 

    In the other direction, we show that $Y(e)$ indeed extends $X$. 
    This is clear in $M_0$, so we focus on a layer $i \in [s]$.
    Let $\lambda_i$ be the result of pushing from $i$ in $\lambda$ 
    and let $C_i$ be the set of crossing edges in $\lambda_i$. 
    Let $V_i'=\lambda_i^{-1}(i)$. 
    By Lemma~\ref{lemma:all-pushing-allowed}, $\lambda_i$ has
    a cost at most as high as $\lambda$, and clearly has at most $k$
    crossing edges, hence $\lambda_i$ is an optimal solution to
    $I$, and by assumption $e \in C_i$. 
    Let $v$ be an endpoint of $e$ such that $\lambda_i(v)\neq i$;
    by assumption, $v$ exists. It follows from the pushing operation
    that $v$ is reachable from $S-t_i$ (respectively
    from $S$ if $i>r$) avoiding $\delta(V_i')$. Thus, by
    Prop.~\ref{prop:extends}, $e'$ extends $\delta(V_i) \cup
    \delta(t_i)$ in $M_i$ (respectively $\delta(V_i)$ if $i>r$), 
    which is the restriction of $X$ to layer $i$. 
    Thus $Y(e)$ extends $X$. 

    Since $\cY^*$ contains at least one set $Y(f)$
    that extends $X$, we conclude $e \in Z_0$. 
  \end{proof}
  
  \noindent
  It only remains to show that sequentially contracting irrelevant
  edges (one at a time, while recomputing $Z_0$ at every step) yields
  a final set $Z$ that contains crossing edges of optimal solutions for
  all metrics as described. This should be clear. The effect of
  contracting an edge $uv$ is precisely to restrict the solution space
  to labellings $\lambda$ where $uv$ is non-crossing. All other edges
  of the graph remain identifiable, and for every instance
  $I$ over $(G,S)$ there still exists some optimal solution $\lambda$
  after the contraction. That is, at every stage, for every instance
  $I$ over $(G,S)$ there exists some solution $\lambda$ with crossing
  number at most $k$ and of cost $\cost(I,k)$, such that the edge that
  is about to be contracted is non-crossing in $\lambda$. Thus the
  final graph, where $Z \cup E_S$ covers the entire edge set, also
  contains such a solution. 
\end{proof}

\begin{corollary}
  For every graph $G$ with terminal set $S$, and integers $k$, $s$, 
  there is a $k$-bounded metric sparsifier $G'$ for $(G,S)$,
  for metrics with up to $s$ labels, with $O(k^{s+1})$ edges,
  which can be computed in randomized polynomial time.
  Furthermore, for every metric $\mu$ on $s$ labels, 
  \textsc{Zero Extension} admits a randomized polynomial kernel with
  $O(q^{s+1})$ edges.
\end{corollary}
\begin{proof}
  The metric sparsifier is constructed by computing the $Z$ as above,
  then contracting all edges not in $Z \cup  E_S$. For the kernel, 
  first assume that $\mu$ has no distinct pair of labels at distance
  zero, as otherwise we can compute a new metric $\mu'$ and instance
  $I'$ by identifying labels and terminals at distance 0. Let $I=(G,
  \tau, \mu, q)$ be an instance. As the metric is integer-valued, it
  follows that any $\tau$-extension $\lambda$ of cost at most $q$
  has crossing number at most $q$. Hence we get a polynomial kernel by
  computing a $q$-bounded metric sparsifier. 
\end{proof}


\section{Tree metrics: Discrete relaxations}
\label{section:tree-metrics}

We now give more powerful algorithms parameterized by the gap parameter 
for problems where the metric embeds into a tree metric. We begin by a
purely combinatorial algorithm for \textsc{Zero Extension} on leaf metrics,
then we move on to the more general \textsc{Zero Extension} and
\textsc{Metric Labelling} problems for general induced tree metrics.
The algorithms for the latter problems rely on 
the \emph{domain consistency} property of the relaxation, 
which allows us to solve the problem by simply branching
on the value of a single variable at a time. 
This property is shown by way of a detour into an analysis of 
properties of VCSP instances whose cost functions are 
\emph{weakly tree submodular}, which is a tractable problem class
containing tree metrics. The algorithms for these problems are then
straight-forward.

At this point, we need to address a subtlety regarding the input cost
function~$\mu$. So far, the cost function only had to obey basic properties
that are easily verifiable or could be seen as a `promise'. However, some of
our arguments below will explicitly need the tree~$T$ that induces the metric.
Luckily this issue has been solved already: given a induced tree metric~$\mu$
over~$D$ in matrix form, one can in time~$O(|D|^2)$ compute a tree that
induces~$\mu$~\cite{TreeMetricConstruction}. If~$\mu$ is a leaf metric, the
output will obviously have~$D$ as the leaves of~$T$. In conclusion,
we will tacitly assume that we have access to the tree~$T$ in the following.

\subsection{Leaf metrics: A duality approach}\label{section:leaf-metric}

\noindent
The \Problem{$\mu$-Edge Disjoint Packing} problem asks, given a graph~$G$
with a terminal set~$S \subseteq V(G)$, to find an edge-disjoint 
packing~$\cP$ of paths whose endpoints both lie in~$S$ that 
maximizes~$\pack(\mu,G,S) := \sum_{P \in \cP} \mu(s_P,t_P)$ (where~$s_P,t_P$ denote
the start- and endpoints of the path~$P$). Hirai and Pap~\cite{SPaths}, as part of a more general 
result, show that if~$\mu$ is a tree metric, then the problem is polynomial-time 
solvable and 
\[
  \pack(\mu,G,S) = \min_{\lambda} \max_{F\subseteq E} \sum_{uv \in E\setminus F} \mu(\lambda(u), \lambda(v)),
\]
where~$\lambda$ is precisely a zero-extension of the terminal-set~$S$ and
the sets~$F \subseteq E$ are so-called \emph{inner odd-join}, that is,
a set of edges whose deletion leaves every non-terminal vertex with an
even degree. It follows that the maximum value of a half-integral
$\tau$-path packing is just the minimum cost of a $\tau$-extension $\lambda$,
since a half-integral path-packing is just a path-packing in the graph
where every edge of $G$ has been duplicated, and such a graph has no
vertices of odd degree. 

Let in the following $I = (G,\tau,\mu,q)$ be an instance of \Problem{Zero Leaf
Extension}, where~$\mu$ is a leaf metric over a tree~$T$ with leaves~$D$. Let
$\hat \mu = \dist_T$ be the underlying tree metric. We define the
\emph{relaxed instance}~$\hat I = (G, \tau, \hat \mu, q)$. Let~$\opt(I)$,
$\opt(\hat I)$ denote the set of optimal solutions for the integral and the
relaxed instance, respectively. As mentioned in Section~\ref{sec:prelims},
it is known that the relaxed instance can be solved optimally in polynomial
time~\cite{ZeroExtOnTrees}.
For convenience, we say that a vertex~$u$ is integral with respect
to a solution~$\lambda$ if~$\lambda(u) \in D$ and we say that an edge~$uv \in G$
is integral with respect to~$\lambda$ if both endpoints are integral.

Using the above notation, we can summarize the duality between 
an minimum relaxed labelling and a path packing as follows:
Given a relaxed instance~$\hat I$, there exists
a half-integral $\tau$-path-packing~$\cP$ of cost precisely $\cost(\hat I)$.
We will not explicitly compute~$\cP$ in the final algorithm, instead we
use its existence to derive useful properties of the problem.
In the following we will use~$S$ to denote the terminals of the instance,
\ie the vertices labelled by~$\tau$. By the usual identification argument, we 
can assume that~$\tau$ is a bijection and a $\tau$-path packing is equivalent
to an~$S$-path packing.

\begin{lemma}\label{lemma:monotone-paths}
  Let~$\cP$ be an half-integral $\tau$-path packing that satisfies
  \[
    \frac{1}{2} \sum_{P \in \cP} \mu(\tau(s_P), \tau(t_P)) = \cost(\hat I).
  \]
  Let~$\lambda \in \opt(\hat I)$ be a relaxed
  optimum and let~$P \in \cP$ with endpoints~$s,t$. Then
  \[
    \cost_{\hat \mu}(\lambda, P) = \mu(\tau(s), \tau(t)).
  \]
\end{lemma}
\begin{proof}
  First, consider \emph{any} $s$-$t$-path~$P'$ in~$G$. Then
  \[
    \cost_{\hat \mu}(\lambda, P) = \sum_{uv \in P'} \hat \mu(\lambda(u), \lambda(v)) \geq \hat \mu(\tau(s), \tau(t)).
  \]
  We therefore find that the inequality
  \[
     \sum_{P \in \cP} \hat \mu(\tau(s_P), \tau(t_P)) 
     \leq \sum_{P \in \cP} \cost_{\hat \mu}(\lambda, P)
     \leq 2\cost_{\hat \mu}(\lambda, G)
  \]
  holds. But according to our assumption, the left-hand side and
  right-hand side are equal and we conclude that
  \[
    \sum_{P \in \cP} \hat \mu(\tau(s_P), \tau(t_P)) 
    = \sum_{P \in \cP} \cost_{\hat \mu}(\lambda, P).
  \]
  and therefore that for every~$P \in \cP$,
  $\cost_{\hat \mu}(\lambda, P) = \mu(\tau(s_P), \tau(t_P))$, as claimed.
\end{proof}

\noindent
A direct consequence is that if we trace an $s$-$t$-path $P \in \cP$,
then the labels assigned by any relaxed optimum~$\lambda$ to~$P$ induce
a monotone sequence from~$s$ to~$t$ in~$T$. That is,
not only will we only encounter those labels that lie on~$T[s,t]$,
we also will encounter them `in order'. We further can
conclude the following:

\begin{corollary}\label{cor:zero-edges}
  Let~$e \in G$ be an edge that is \emph{not} part of the path-packing $\cP$.
  Then under every relaxed optimum~$\lambda \in \opt(\hat I)$ the edge~$e$ has cost
  zero.
\end{corollary}

\noindent
Consider an edge~$xy \in E(T)$. Then, as a consequence of 
Lemma~\ref{lemma:monotone-paths} the set of edges 
$C_{xy}(\lambda) =\{uv \in E(G) \mid \lambda(u) \in T_x, \lambda(v) \in T_y\}$
between the vertex sets with labels in~$T_x$ and~$T_y$, respectively,
must be saturated by paths of the packing~$\cP$. For cuts right
above leafs of~$T$, this implies the following.

\begin{lemma}\label{lemma:packing-saturates-cuts}
  Let~$S$ be the vertices labelled by~$\tau$ in~$G$ and assume that
  $\tau$ is a bijection.
  Let~$C$ be any minimum $(x, S-x)$-cut for some terminal~$x \in S$. Then
  every optimal, half-integral $S$-path-packing in $G$ will saturate $C$.
\end{lemma}
\begin{proof}
  Let us first consider the closest minimal cut~$\closest(x)$ and the furthest
  minimal cut~$\furthest(x)$. Let $\cP$ be a max-value $\tau$-path-packing, and let
  $\lambda$ be the corresponding min-cost extension of $S$. Let $y$ be the
  ancestor of~$x$ in~$T$ and consider $C_{yx}(\lambda)$. By the above, $C_{yx}(\lambda)$
  is saturated by $\cP$. Since every path of~$\cP$ induces a monotone sequence in
  $T$ under $\lambda$, every path~$P \in \cP$ crossing~$yx$ in $T$ must
  have~$x$ as an endpoint. But the total weight of such paths~$\cP_x \subseteq
  \cP$ is at most~$\closest(x)$.  Since $C$ is a $(x, S-x)$-cut, we must have
  $C_{yx}(\lambda) = \closest(x)$. 

  Now, since~$\furthest(x)$ is a minimal $(x, S-x)$-cut as well, $\cP_x$
  saturates~$\furthest(x)$ as well. This in particular means that every path
  in~$\cP$ that intersects~$\furthest(x)$ must end in~$x$, \eg those paths are
  exactly~$\cP_x$. Consequently, every min-cut around~$x$
  is saturated by~$\cP_x$, proving the statement.
\end{proof}

\begin{lemma}\label{lemma:cost-by-cuts}
  Let~$\lambda$ be a not necessarily optimal solution for~$\hat I$. Then
  \[
    \cost_\mu(\lambda, G) = \sum_{xy \in T} |C_{xy}(\lambda)|.
  \]
\end{lemma}
\begin{proof}
  The claim is equivalent to proving that 
  \[
    \sum_{uv \in G} \mu(\lambda(u), \lambda(v)) = \sum_{xy \in T} |C_{xy}(\lambda)|.
  \]
  We proof the above equality by double-counting. Consider an edge~$uv$, then
  $\mu(\lambda(u), \lambda(v))$ is by definition~$|T[\lambda(u), \lambda(v)]|$.
  Note that~$uv$ appears in exactly those cuts~$C_{xy}$ with~$xy \in T[\lambda(u),\lambda(v)]$,
  hence we can charge the cost of~$\mu(\lambda(u), \lambda(v))$ on those cuts
  and the equation follows.
\end{proof}

\begin{lemma}\label{lemma:no-cut-branch}
  Let~$u \in G$ be such that~$u \not \in \furthest(x)$ for all terminals~$x
  \in S$. Then~$u$ does not receive any integral value by any relaxed optimum.
\end{lemma}
\begin{proof}
  Assume otherwise: let~$\lambda \in \opt(\hat I)$ be a relaxed optimum that assigns~$u$
  some integral value~$x \in T$. Let~$y$ be the parent of~$x$ in~$T$. Then we conclude
  that~$C_{xy}(\lambda)$ cannot be a minimum cut, since~$u$ is part of~$\lambda^{-1}(x)$,
  but~$u$ is not in the furthest min-cut~$\furthest(x)$. But then~$\sum_{xy \in T} |C_{xy}(\lambda)|$
  cannot be minimal, and by Lemma~\ref{lemma:cost-by-cuts} therefore~$\cost_\mu(\lambda, G)$
  is not either, contradicting our assumption.
\end{proof}

\begin{lemma}\label{lemma:three-cuts-empty}
  For every triple~$x,y,z \in S$ of distinct terminals it holds that
  $\furthest(x) \cap \furthest(y) \cap \furthest(z) = \emptyset$.
\end{lemma}
\begin{proof}
  Assume towards a contradiction that there exists~$u \in \furthest(x) \cap
  \furthest(y) \cap \furthest(z)$ and let us choose~$u$ such that it is incident
  to edges that cross the cut~$\furthest(x)$.

   Let~$\cP$ be an optimal half-integral $S$-path packing. By Lemma~\ref{lemma:packing-saturates-cuts}, 
   all paths of~$\cP$ that enter~$\furthest(c)$ for~$c
  \in S$ must have an endpoint in~$c$; moreover, they saturated the cut~$\furthest(c)$. 
  Since~$x$ is incident to edges that cross the cut~$\furthest(x)$ it must therefore lie
  on at least one path~$P$ that ends in~$x$. Now, since~$u$ also lies in~$\furthest(y)$,
  the path~$P$ crosses~$\furthest(y)$ and must therefore end in~$y$. However, the same
  argument holds for~$z$ and we arrive at a contradiction. We conclude that the
  intersection of the three cuts must indeed be empty. 
\end{proof}

\thmleaves*
\begin{proof}
  Given the input graph~$G$ we first construct for every edge~$ij \in T$ a
  flow network~$H_{ij}$ from~$G$ as follows: let~$D_i$ be those leaves that
  lie in the same component as~$i$ in~$T-ij$ and~$D_j$ all others.
  Then~$H_{ij}$ is obtained from~$G$ by identifying all
  terminals~$\tau^{-1}(D_i)$ into a source~$s$ and all terminals~$\tau^{-1}(D_j)$ into a
  sink~$t$. For each networks~$H_{ij}$ we compute a maximum flow~$f_{ij}$ in
  time~$\phi(n,m)$. By Lemma~\ref{lemma:cost-by-cuts} we have that for
  every~$\lambda \in \opt(\hat I)$ it holds that
  \[
    \sum_{xy \in T} |f_{ij}| = \sum_{xy \in T} |C_{xy}(\lambda)| = \cost_\mu(\hat I).
  \]
  Note that we can also, in linear time, find the closest cuts~$\closest(x)$ and
  furthest cuts~$\furthest(x)$ for terminals~$x \in S$ using the 
  residual network of~$(H_{ij}, f_{ij})$ with~$i = \tau(x)$ and~$j$ the parent
  of~$i$ in~$T$.

  Thus, in linear time, we can identify whether~$G$ contains a vertex that is
  not part of any furthest min-cut~$\furthest(x)$ for all~$x \in S$. By
  Lemma~\ref{lemma:no-cut-branch}, such a vertex cannot take an integral value
  in any relaxed optimum. We therefore branch on the~$|D|$ possible integral
  values it could take: for~$x \in S$ with~$\tau(x) \in D$ being the chosen
  integral value, we update the flow networks~$(H_{ij},f_{ij})$ by adding an
  edge of infinite capacity from~$x$ to~$u$ and then augment the flow~$f_{ij}$
  until it is maximum again. Note that the number of augmentations necessary
  are at most $k-p$, since each augmentation witnesses the increase
  of~$p$ and thus the decrease of the parameter by one. In our analysis we
  can therefore charge each augmentation to a level of the search tree
  (treating each augmentation like a descent to the next node) and thus spend
  only~$O(m)$ time per flow~$f_{ij}$, for a total of~$O(|T|m)$.

  Otherwise, we find that every vertex of the current graph~$G$ is contained
  in at least one furthest min-cut. By Lemma~\ref{lemma:three-cuts-empty}, the
  intersection of three or more such cuts is empty and we can partition the
  vertices of~$G$ into sets~$\{V_x\}_{x \in S} \cup \{U_{xy}\}_{x,y \in S}$
  where~$V_x$ contains all vertices that are only contained in~$\furthest(x)$
  while~$U_{xy}$ contains those that live in the intersection~$\furthest(x) \cap
  \furthest(y)$. By Lemma~\ref{lemma:packing-saturates-cuts} we have
  that~$U_{xy}$ only contains edges towards~$V_x$ and~$V_y$ (since these edges
  are exactly saturated by a half-integral path-packing and the paths saturating
  these edges have endpoints $x$ and~$y$). Therefore we construct an integral
  solution~$\lambda$ as follows: every set~$V_x$, $x \in S$ is coloured~$\tau(x)$
  and for every non-empty set~$U_{xy}$, $x,y \in S$ we choose colour~$\tau(x)$
  or~$\tau(y)$ arbitrarily. By Lemma~\ref{lemma:cost-by-cuts}, the cost of~$\lambda$
  is precisely
  \[
    \cost_\mu(\lambda, G) = \sum_{xy \in T} |C_{xy}(\lambda)| = \sum_{xy \in T} |f_{ij}| = \opt(\hat I)
  \]
  and we conclude that~$\lambda$ is an integral solution that matches the
  relaxed optimum of the current instance. In this case, we return~$\lambda$ 
  as a solution to the original instance.
  The claimed running time follows if we prune every branch of the
  search tree in which the parameter drops below zero.
\end{proof}

\subsection{VCSP toolkit}

\noindent
We now review some required terminology and tools for the proof of the
algebraic properties of distance problems on trees. 

Given a set of cost functions $\Gamma$ over a domain $D$,
an instance $I$ of VCSP$(\Gamma)$ is defined by a set of variables $V$
and a sum of valued constraints $f_i(\bar v_i)$,
where for each $i$, $f_i \in \Gamma$ and $\bar v_i$ is a tuple
of variables over $V$. We write $f_i(\bar v)\in I$ to signify
that $f_i(\bar v)$ is a valued constraint in $I$. 

It is known that the tractability of a VCSP is characterized by
certain algebraic properties of the set of cost functions. 
In full generality, such conditions are known as \emph{fractional
  polymorphisms} for the finite-valued case and more general
\emph{weighted polymorphisms} in the general-valued case.
Dichotomies are known in these terms both for the
finite-valued~\cite{ThapperZivnyCSP} and general case of VCSP~\cite{KolmogorovKR17},
i.e., characterizations of each VCSP as being either in P or NP-hard.
We will only need a less general term. 

A \emph{binary multimorphism}~$\angled{\circ, \bullet}$ of a language~$\Gamma$
over a domain~$D$ is a pair of binary operators that satisfy
\[
   f(\bar x) + f(\bar y) \geq 
   f(\bar x \circ \bar y) + f(\bar x \bullet \bar y)
   \qquad \forall f \in \Gamma, \bar x, \bar y \in D^{\ar(f)},
\]
where $\ar(f)$ is the arity of $f$ and 
where we extend the binary operators to vectors by applying them
coordinate-wise. 
An operator $\circ$ is \emph{idempotent} if $x \circ x = x$
for every $x \in D$, and \emph{commutative} if $x \circ y = y \circ x$. 
A (finite, finite-valued) language $\Gamma$ with a binary
multimorphism where both operators are idempotent and commutative is
solvable in polynomial time via an LP-relaxation~\cite{ThapperZivnyCSP}.
The most basic example is the Boolean domain $D=\{0,1\}$,
in which case the multimorphism $\angled{\land, \lor}$ corresponds to
the well-known class of \emph{submodular functions}, which is a
tractable class that generalizes cut functions in graphs. 


The following is folklore, but will be important to our investigations. 
Again, the corresponding statements apply for arbitrary fractional 
polymorphisms, but we only give the version we need in the present paper. 


\begin{definition}[Preserved under equality]
  Let $f$ be a function that admits a multimorphism~$\angled{\circ, \bullet}$.
  We say that two tuples~$\bar x, \bar y \in D^{\ar(f)}$ are 
  \emph{preserved under equality} if
  \[
  f(\bar x) + f(\bar y) = f(\bar x \circ \bar y) + f(\bar x \bullet \bar y).
  \]
  For a relation $R \subseteq D^{\ar(r)}$,
  we say that $f$ is \emph{preserved under equality in $R$}
  if every pair of tuples $\bar x, \bar y \in R$ is preserved under equality and 
  $\bar x \circ \bar y, \bar x \bullet \bar y \in R$.
\end{definition}


\begin{lemma}\label{lemma:opt-local}
  Let~$\Gamma$ be a language of cost functions that admit a
  multimorphism~$\angled{\circ, \bullet}$ and let~$\lambda_1, \lambda_2 \in \opt(I)$ 
  for some instance~$I$ of \emph{VSCP}$(\Gamma)$. Then for every
  valued constraint~$f(\bar v) \in I$ it holds that
  \begin{align*}
    f(\lambda_1(\bar v)) + f(\lambda_2(\bar v))
    = f((\lambda_1 \circ \lambda_2)(\bar v))
    + f((\lambda_1 \bullet \lambda_2)(\bar v)),
  \end{align*}
  where $f(\lambda(\bar v))=f(\lambda(v_1), \ldots, \lambda(v_r))$ for
  $\bar v=v_1, \ldots, v_r$ is the value of $f(\bar v)$ under $\lambda$.
  In other words, every valued constraint~$f(\bar v) \in I$ is
  preserved under equality in $\opt(I)$.
\end{lemma}
\begin{proof}
  Let~$\cost_f(\lambda)$ be the sum of all valued constraints~$f(\bar v) \in I$ under~$\lambda$.
  By the multimorphism, we have that
  \[
    \cost_{f}(\lambda_1) + \cost_{f}(\lambda_2) \geq 
    \cost_{f}(\lambda_1 \circ \lambda_2) + \cost_{f}(\lambda_1 \bullet \lambda_2)
  \]
  and since~$\lambda_1$ and~$\lambda_2$ are optimal we obtain that
  \[
    \cost_{f}(\lambda_1) + \cost_{f}(\lambda_2) = 2\opt(\hat I) =
    \cost_{f}(\lambda_1 \circ \lambda_2) + \cost_{f}(\lambda_1 \bullet \lambda_2)
  \]
  and therefore that
  $
    \cost_{f}(\lambda_1 \circ \lambda_2) = \cost_{f}(\lambda_1 \bullet \lambda_2) = \opt(\hat I)
  $.
  For two variables~$u,v$ that appear together in a valued constraint~$f$ let us define
  \[
    \Delta_{f}(\bar v) := f(\lambda_1(\bar v)) + f(\lambda_2(\bar v))
    - f( (\lambda_1 \circ \lambda_2)(\bar v))
    - f((\lambda_1 \bullet \lambda_2)(\bar v)),
  \]
  then by the multimorphism property it follows that~$\Delta_{f}(\bar v) \geq 0$. 
  Since, by definition,
  \begin{align*}
      \cost_{f}(\lambda_1) + \cost_{f}(\lambda_2)
    - \cost_{f}(\lambda_1 \circ \lambda_2) - \cost_{f}(\lambda_1 \bullet \lambda_2) 
    = \sum_{f(\bar v) \in I} \Delta_f(\bar v)
  \end{align*}  
  and the left-hand side evaluates to zero, we conclude that~$\sum_{uv \in G} \Delta_{f}(u,v) = 0$
  and therefore that~$\Delta_{f}(u,v) = 0$ for every constraint~$f(u,v) \in I$.
\end{proof}

\noindent
To illustrate, let us return again to the case of graph cut functions
and submodularity over the Boolean domain. Let $G=(V,E)$ be an
undirected graph, and define the cut function 
$f_G \colon 2^V \to \Z$ as $f_G(S)=|\delta(S)|$. 
Then $f_G$ is the sum over binary valued constraints
$f(u,v)=[u \neq v]$ over all edges $uv \in E$, in Iverson bracket notation.
Since a single valued constraint $f(u,v)$ is submodular, the same
holds for the cut function as a whole. 
Then Lemma~\ref{lemma:opt-local} specialises into the statement 
that for two sets $A, B \subset V$ such that $\delta(A), \delta(B)$
are minimum $s$-$t$-cuts in $G$ for some $s, t \in V$, 
there is no edge between $A \setminus B$ and $B \setminus A$. 
This kind of observation is a common tool in, e.g., graph theory 
and approximation algorithms.

The above lemma will be very useful when reasoning about the structure
of $\opt(I)$ subject to more complex multimorphisms, as we will define
next. 

\subsection{Submodularity on trees}

\noindent
Let~$\preceq_T$ denote the ancestor relationship in a rooted tree~$T$.
For a path~$P[x,y] \subseteq T$, let~$z_1,z_2$ be the middle vertices of~$P[x,y]$
(allowing~$z_1 = z_2$ in case~$P[x,y]$ has odd length) such that~$z_1 \preceq_T
z_2$. Define the commutative operators~$\midup, \middown$
as returning exactly those two mid vertices, \eg 
$x \midup y = y \midup x = z_1$ and
$x \middown y = y \middown x = z_2$. Languages admitting the 
multimorphism $\angled{\midup, \middown}$ are called 
\emph{strongly tree-submodular}.

Define the commutative operator~$\lca$ to return the common ancestor of two
nodes~$x,y$ in a rooted tree~$T$. Define~$x \lcaskew y$ to be the vertex~$z$
on~$P[x,y]$ which satisfies~$\dist_T(x,z) = \dist_T(y, x \lca y)$. In other
words, to find~$z = x \lcaskew y$, we measure the distance from~$y$ to the
common ancestor of~$x$ and~$y$ and walk the same distance from~$x$
along~$P[x,y]$. Languages that admit~$\angled{\lca,\lcaskew}$ as a
multimorphism are called \emph{weakly tree-submodular}. In particular, all
strongly tree-submodular languages are weakly tree-submodular~\cite{TreeSubmodular}.
Tree-metric are, not very surprisingly, strongly tree-submodular:

\begin{lemma}\label{lemma:tm-strong}
  Every tree-metric is strongly tree-submodular for every rooted
  version of the tree.
\end{lemma}
\begin{proof}
  Let~$T$ be a rooted tree and let~$a,b,x,y \in T$ not necessarily distinct nodes. 
  We let~$\dist$ be the distance-metric on~$T$. We need to show that
  \begin{align}
    \dist(a,b) + \dist(x,y) \geq \dist(a \midup x, b \midup y) + \dist(a \middown x, b \middown y). \label{eq:treesub1}
  \end{align}
  First, consider the case that~$P[a,b] \cap P[x,y] = P[m_1,m_2]$ is a
  non-empty path. Assume that~$P[a,x]$ and~$P[b,y]$ are disjoint, Then the 
  left-hand side of (\ref{eq:treesub1}) is equal to
  \begin{align}
    \dist(a,b) + \dist(x,y) = |P[a,b]| + |P[x,y]| = 2|P[m_1,m_2]| + |P[a,x]| + |P[b,y]|. \label{eq:treesub2}
  \end{align}
  Since the nodes~$a \midup x$, $a \middown x$ both lie on~$P[a,x]$ and the
  nodes $b \midup y$, $b \middown y$ on~$P[a,y]$, the right-hand side of
  (\ref{eq:treesub1}) cannot be larger than the right-hand side of
  (\ref{eq:treesub2}), thus (\ref{eq:treesub1}) holds in this case.
  In the alternative case where~$P[a,x]$ and~$P[b,y]$ are non-disjoint, we
  instead use the paths~$P[a,y]$, $P[b,x]$. In this case, the 
  nodes~$a \midup x$, $a \middown x$, $b \midup y$, $b \middown y$ 
  could now also lie on~$P[m_1,m_2]$ but the argument remains the same.

  Thus consider the second case: $P[a,b]$ and~$P[x,y]$ do not intersect.
  Let now~$P[c,z]$ be the unique path connecting~$P[a,b]$ and~$P[x,y]$
  with~$c \in P[a,b]$ and~$z \in P[x,y]$. First, we simplify our lives 
  by observing that the right-hand side of (\ref{eq:treesub1}) can be
  replaced using
  \[
    \dist(a \midup x, b \midup y) + \dist(a \middown x, b \middown y) = 2\dist(m_{ax}, m_{by})
  \]
  where we allow the mid-points~$m_{ax}$ of~$P[a,x]$ and $m_{by}$ of~$P[b,y]$
  to lie in the middle of an edge (by some abuse of notation we
  extend~$\dist$ to such mid-points of edges and allow it to take 
  half-integral values). Consider the following re-writing of (\ref{eq:treesub1}):
  \begin{align}
    \dist(a,b) + \dist(x,y) \geq 2\dist(m_{ax}, m_{by}). \label{eq:treesub3}
  \end{align}
  Clearly, it holds in the degenerate case of~$a = b = c$ and~$x = y = z$.
  We prove the remainder by induction through the insertion of an arbitrary edge.
  First, assume that an edge is inserted into~$P[a,c]$. This increases the left-hand
  side of (\ref{eq:treesub3}) by one and moves~$m_{ax}$ by half a unit, thus at most
  increasing the right-hand by one as well. The same holds, by symmetry, for any
  edge inserted into~$P[b,c]$, $P[x,z]$, and~$P[y,z]$. It remains to consider edges
  inserted into~$P[c,z]$ whose addition does not contribute to the left-hand side.
  If such an edge additionally lies on the path between~$m_{ax}$ and~$m_{by}$, the distance
  between these two mid-points decreases by one; otherwise both midpoints are shifted
  in such a way that they remain equidistant. In neither scenario does the right-hand
  side increase, proving the claim. Finally, since $\dist$ does not
  depend on the choice of root in $T$, the result also holds for every
  root.
\end{proof}

\begin{corollary}\label{cor:tm-weak}
  Every tree-metric is weakly tree-submodular for every rooted
  version of the tree.
\end{corollary}

\noindent
We will need the following characterization of which value-pairs are 
preserved under equality by strong tree submodularity
for tree distance functions. 

\begin{lemma}\label{lemma:tree-collinear}
  Two tuples~$(a,b),(x,y) \in V(T) \times V(T)$ are preserved under equality
  by~$\dist_T$ with multimorphism~$\angled{\midup, \middown}$ iff all four nodes lie on a single
  path~$P$ in~$T$ \emph{and} either~$a,b \leq_P x,y$ or~$a,x \leq_P b,y$.
\end{lemma}
Below is a complete enumeration of all orders in which the nodes~$a,b,x,y$
might appear on a path, where we removed all cases in which $b$ appears before
$a$ (to break mirror symmetry) and all cases derivable by exchanging $a$
with~$x$ and~$b$ with~$y$, or applying both operations. The cases on the right
side show how the nodes might all appear on a single path and yet not be
preserved under quality by~$\angled{\midup, \middown}$ (if they were, the sum
of the magenta lines would match the sum of the cyan lines)
\begin{center}
  \includegraphics[scale=.33]{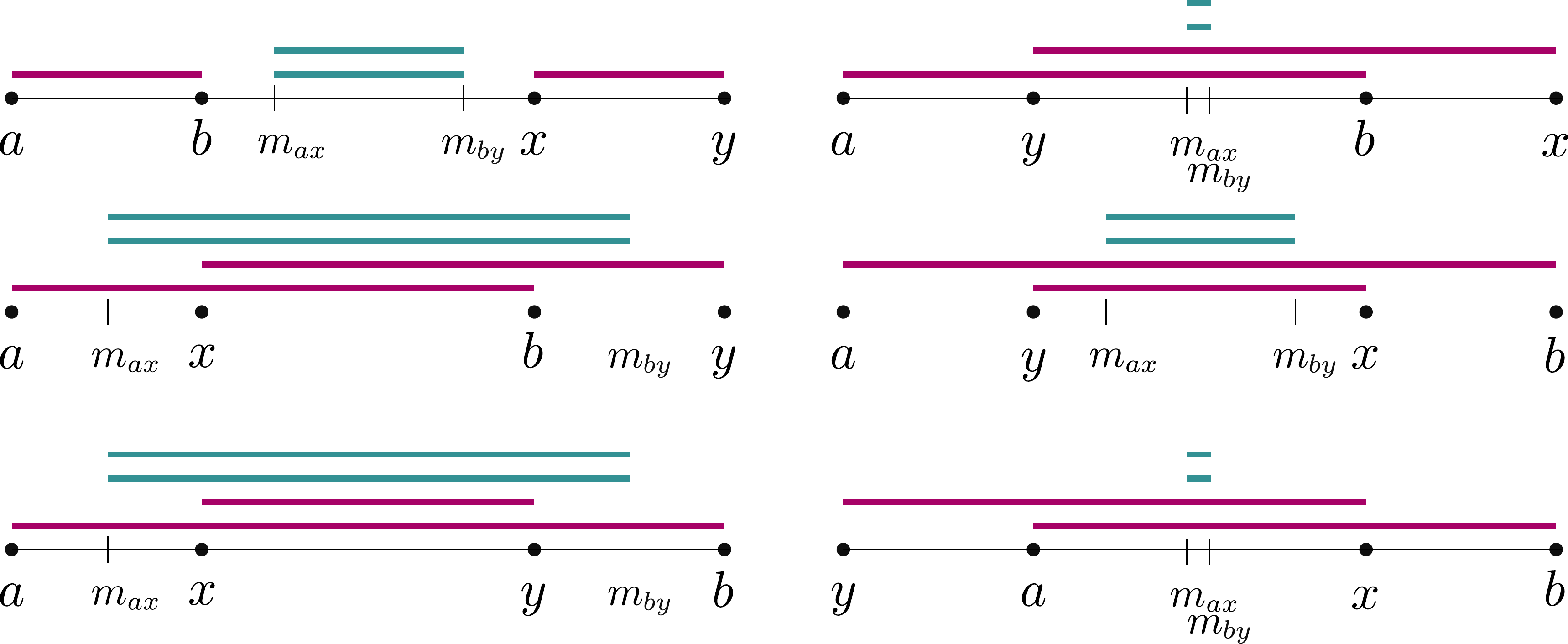}
\end{center}

\begin{proof}
  Let~$m_{ax}$ and~$m_{bx}$ denote the mid-points of~$P[a,x]$ and~$P[b,y]$, as
  in the proof of Lemma~\ref{lemma:tm-strong} we allow these points to lie in
  the middle of an edge in case these paths are of even length. We drop the
  subscript of~$\dist_T$ in the following.

  For the one direction, assume that~$a,b,x,y$ all lie on some common path~$P$.
  We now need to show that
  \begin{align}
    \dist(a,b) + \dist(x,y) = 2\dist(m_{ax}, m_{by}). \label{eq:collin1}
  \end{align}
  It will be helpful to identify~$P$ with the interval~$I = [0,|P|]$. 
  We distinguish several cases depending on the order imposed on these nodes by~$P$,
  the two principal cases are depicted below. Since~$a,b$ are exchangable, we will
  assume that~$a \leq_P b$ in all cases.
  \begin{center}
    \includegraphics[scale=.33]{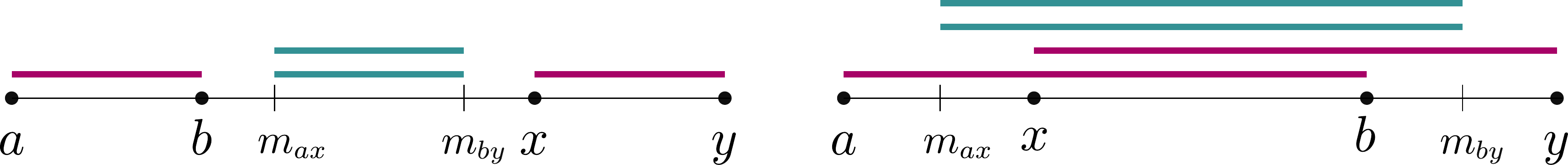}
  \end{center}
  In both of these pictures we can easily verify that
  $m_{ax}$ lies at position~$\dist(a,x)/2$ on~$I$ 
  and~$m_{by}$ at~$\dist(a,b) + \dist(b,y)/2$, hence
  \begin{align*}
    2\dist(m_{ax},m_{by}) &= 2\Big| \frac{\dist(a,x)}{2} - \big(\dist(a,b) + \frac{\dist(b,y)}{2}\big) \Big| 
                           = \Big| \! \dist(a,x) - \dist(b,y) - 2\dist(a,b) \Big| \\
                          &= \dist(a,b) + \dist(x,y).
  \end{align*}  
  Note that the case~$axyb$ is equivalent to the case 
  $axby$, since exchanging~$y$ and~$b$ does not affect $m_{by}$
  and thus none of the relevant distances. This proves the first direction.

  In the other direction, first assume that the nodes do all lie on a path~$P$
  but do no fulfil the second property. After removing symmetries we are left
  with the four cases~$aybx$, $ayxb$, $yaxb$. Consider~$aybx$ first. We have that
  \begin{align*}
    2\dist(m_{ax},m_{by}) &= 2\Big| \frac{\dist(a,x)}{2} - \big( \dist(a,y) + \frac{\dist(y,b)}{2} \big) \Big|
                           =  \Big| \! \dist(a,x) - \dist(y,b) - 2\dist(a,y) \Big| \\
                          &=  \Big| \! \dist(a,y) + \dist(b,x) - 2\dist(a,y) \Big| 
                           =  \Big| \! \dist(b,x) - \dist(a,y) \Big|,
  \end{align*}
  which is smaller than either~$\dist(a,b)$ or~$\dist(x,y)$. The calculation for
  $yaxb$ is essentially the same, which leaves us with~$ayxb$. In that case (skipping to
  the part where the computation diverges) we obtain that
  \begin{align*}
    2\dist(m_{ax},m_{by}) &=  \Big| \! \dist(a,x) - \dist(y,b) - 2\dist(a,y) \Big| \\
                          &=  \Big| \! \dist(a,y) + \dist(y,x) - \dist(y,x) - \dist(x,b) - 2\dist(a,y) \Big| \\
                          &=  \dist(a,y) + \dist(x,b),
  \end{align*}
  which is smaller than~$\dist(a,b)$ unless $x = y$.

  This concludes the case in which the nodes all lie on a path, hence we are
  left with cases in which the nodes~$a,b,x,y$ do not lie on a single path
  in~$T$. First, consider the case $aybx$ (as in the above figure) and imagine
  introducing edges to take~$y$ and~$b$ away from the path~$P[a,x]$. Every
  edge introduced in this manner will contribute exactly one the the left-hand
  side of (\ref{eq:collin1}) and at most one to the right-hand side  (since
  $m_{by}$ moves by half a unit). Hence no tree derivable from~$aybx$ can ever
  achieve equality. The same argument holds for~$ayxb$ (where we introduced
  edges to remove~$y$ and~$x$ from~$P[a,b]$) and $yaxb$ (were we remove~$a$
  and~$x$ from~$P[y,b]$).

  In the case of~$axby$, removing~$x$ or~$b$ from~$P[a,y]$ by introducing an
  edge will contribute one to the left-hand side of~(\ref{eq:collin1}) and
  decrease the right-hand side (since~$m_{ax}$ would move away from~$a$ and~$m_{by}$
  away from~$y$), hence equality is broken. A similar argument works for
  $axyb$ and~$abxy$, in both cases the midpoints move in a way that decreases
  the right-hand side of~(\ref{eq:collin1}). This proves the claim.
\end{proof}

\begin{corollary} \label{corollary:edge-supports}
  Let $d_T$ be preserved under equality in $R$ for some $R \subseteq
  V_T \times V_T$, with at least one pair $(a,b) \in R$ with $a \neq b$.
  Then there is a path $P$ in $T$ which 
  can be oriented as a directed path such that for every pair $(a,b) \in R$
  the nodes $a$ and $b$ lie on $P$ with $a \preceq_P b$. 
\end{corollary}
\begin{proof}
  We first show that it holds for all pairs $(a,b) \in R$ with $a \neq b$. 
  Let $ij$ be an edge of $T$ and let $T_i$, $T_j$ be the trees of $T-ij$.
  By the above lemma, there are no two pairs $(a,b), (x,y) \in R$
  such that $a, y \in T_i$ and $b, x \in T_j$. Hence we can define an
  oriented subforest $T'$ of $T$ by including a directed edge $ij$ 
  whenever there is a pair $(a,b) \in R$ with $a \in T_i$, $b \in T_j$.
  Then again by the above lemma, there is no path $P$ in $T$ such that 
  $T'$ contains edges of $P$ oriented in conflicting directions. 
  This implies that $T'$ is a subgraph of a directed path in $T$.
  
  Next, let $P$ be a minimal directed path as above, i.e., $P=T[s,t]$
  for some $s, t$ such that $R$ contains pairs $(s, s'), (t', t)$ for
  some $s', t'$, and all pairs of $R$ of non-zero length lie on $P$. 
  Let $(a,a) \in R$ be a point not on $P$ and let $b$ be the point
  on $P$ closest to $a$. By Lemma~\ref{lemma:tree-collinear},
  there is no pair $(x,y) \in R$ with $x \prec b \prec y$. 
  Let $(x,y) \in R$ with $x \neq y$, and assume that $x \prec y \prec b$.
  We claim that by iterating the strong tree submodularity
  multimorphism, starting with $(x,y)$ and $(a,a)$, we
  can generate new pairs $(x',y') \in R$ with $x' \neq y'$, such that
  $x'$ and $y'$ are strictly closer to $a$ than $x$ and $y$. Indeed,
  let $m_{xa}$ and $m_{ya}$ be the midpoints, which as above are
  allowed to lie on the middle of edges. Then $m_{xa}$ and $m_{ya}$
  are distinct, and if they are at distance at least 1 the conclusion
  is clear. Otherwise one of the points is integral, and one of the 
  operations $\midup$, $\middown$ rounds the other point in the other
  direction, generating a pair $(x', y')$ of non-zero length. By
  iterating this, we eventually generate a pair $(b,b') \in R$ where
  $b \neq b'$ and $b'$ does not lie on $P$, which is a contradiction.
\end{proof}

\subsection{The domain consistency property}

Consider a problem VCSP$(\Gamma)$ over a domain $D_I$
and a discrete relaxation VCSP$(\Gamma')$ of VCSP$(\Gamma)$
over a domain $D \supseteq D_I$. We say that the relaxation
has the \emph{domain consistency property} if the following holds:
for any instance $I$ of VCSP$(\Gamma')$, if for every variable $v$
there is an optimal solution to $I$ where $v$ takes a value in $D_I$, 
then there is an optimal solution where all variables 
take values in $D_I$, \ie an optimal solution to the original
problem of the same cost. 
We will use the results of the previous section to show that the
relaxations we are using for \textsc{Zero Extension} and
\textsc{Metric Labelling} on induced tree metrics have the domain
consistency property (in the latter case with a suitable 
restriction on the unary costs), allowing for $\FPT$ algorithms under the 
gap parameter via simple branching algorithms. 

The result will follow from a careful investigation of the binary 
constraints that $\opt(I)$ can induce on a pair of vertices $u, v \in V$,
also for cases when there is no edge $uv$ in $G$. 
The result builds on Corollary~\ref{corollary:edge-supports}.

For the rest of the section, let us fix a relaxed instance $I=(G=(V,E), \tau,\mu,q)$
of \textsc{Zero Extension} where $\mu$ is a tree metric defined by a tree $T$, 
and the original (non-relaxed) metric is the restriction of $\mu$
to a set of nodes $D_I$.
Note that $I$ can be expressed as a VCSP instance using assignments
and the valued constraint $\mu$. 
Let $\opt$ be the set of optimal labellings. 
For a vertex $v \in V$, let $D(u)$ denote the set
$\{\lambda(u) \mid \lambda \in \opt\}$, and let $D_I(v)=D_I \cap D(v)$. 
Furthermore, for a pair of vertices $u, v \in V$, 
let $R(u,v)=\{(\lambda(u), \lambda(v)) \mid \lambda \in \opt\}$
be the projection of~$\opt$ onto~$(u,v)$,
and $R_I(u,v)=R(u,v) \cap (D_I \times D_I)$ the integral part of this
projection. Let $F \subseteq E$ be the set of edges that are crossing in at least
one $\lambda \in \opt$ and let $E_0 = E \setminus F$.  

We begin by observing that the ``path property'' of
Corollary~\ref{corollary:edge-supports} applies to all 
vertices and edges in $\opt$. 

\begin{lemma} \label{lemma:path-support-everywhere}
  For every vertex $v$ that lies in a connected component of $G$
  containing at least one terminal, $D(v)$ is a path in $T$. 
  Furthermore, for every edge $uv \in E$, $R(u,v)$ embeds
  into the transitive closure of a directed path in $T$. 
\end{lemma}
\begin{proof}
  First note that $v$ has the same domain as any vertex $u$ reachable
  from $v$ via edges of $E_0$. Therefore we can focus on the case that
  $v$ is incident with an edge $uv \in F$. In this case, the 
  valued constraint $\dist_T(u, v)$ is present in $I$, 
  and by Lemma~\ref{lemma:opt-local} is preserved under equality in
  $R(u,v)$. Since $D(v)$ is just the projection of $R(u,v)$ to $v$,
  it follows from Corollary~\ref{corollary:edge-supports} that
  $D(v)$ is contained in a path. Furthermore, $D(v)$ itself is closed
  under the $\midup$, $\middown$ operations, which implies that
  $D(v)$ covers the whole path. 
  For the second part, if $uv \in F$ then the claim is 
  Corollary~\ref{corollary:edge-supports}.
  Otherwise, $R(u,v)$ is a collection of pairs $(x,x)$
  for all $x \in D(v)$, which is equal to a path in $T$. 
\end{proof}

\noindent
Next, we show the main result of this section: if $u$ and $v$ is a
pair of variables, then whether or not there is an edge $uv$ in $E$, 
the constraint $R(u,v)$ induced on $u$ and $v$ by $\opt$ is only
non-trivial on values in $D(u) \cap D(v)$. 
Note that the proof only uses the algebraic properties of weak tree
submodularity, hence the only assumption that is specific to
\textsc{Zero Extension} is that $D(u)$ and $D(v)$ form paths in $T$.

\begin{lemma} \label{lemma:all-non-shared}
  Let $u$ and $v$ be a pair of variables and $a \in D(u)$, 
  $b \in D(v)$ a pair of values. If $(a,b) \notin R(u,v)$,
  then $a, b \in D(u) \cap D(v)$ and $a \neq b$. 
\end{lemma}
\begin{proof}
  Refer to the values of $D(u) \cap D(v)$ as \emph{shared values}, and
  other values of $D(u)$ and $D(v)$ as \emph{non-shared values}. 
  Since $D(u)$ and $D(v)$ is each a path in $T$,
  $R(u,v)$ is contained in the product of these paths,
  and the shared values (if any) induce a common path. 
  Let $P_u$ be the path on vertices $D(u)$ and $P_v$ the path
  on vertices $D(v)$. Let $\omega_d$ denote the $d$-rooted weak tree
  submodularity multimorphism. Using the operation $\lca$ rooted in $d$,
  we find for any $d$ that $R(u,v)$ contains a pair $(a,b)$ where
  $a \in D(u)$ and $b \in D(v)$ are both chosen as the point closest to
  $d$. We refer to this as the \emph{$d$-closest pair}. 
  We first handle the case that there are no shared values.

  \begin{claim}
    If there are no shared values, then $R(u,v)=D(u) \times D(v)$.
  \end{claim}
  \begin{proof}
    Let $p \in P_u$ and $q \in P_v$ be the points on the respective
    path closest to the other path. By considering the
    $p$-closest pair in $R(u,v)$, we find $(p,q) \in R(u,v)$. 
    Now let $a \in D(u)$ and $b \in D(v)$ be arbitrary. 
    By considering the $a$-closest and $b$-closest pairs in $R(u,v)$,
    we find $(a,q), (p,b) \in R(u,v)$. But then $\omega_p$
    on this pair produces $(a,b) \in R(u,v)$. 
  \end{proof}

  \noindent
  Next, assume that there are shared values, and let $P_I$ be the path
  induced on these values. Note that $(r,r) \in R(u,v)$ for every $r
  \in V(P_I)$, by considering the $r$-closest pair.  
  We also have the following. 

  \begin{claim}
    Every vertex $r \in D(u) \cap D(v)$ is compatible with
    every vertex of $D(u) \symdiff D(v)$, i.e., 
    $(a,r) \in R(u,v)$ for every $a \in D(u) \setminus D(v)$,
    and $(r,b) \in R(u,v)$ for every $b \in D(v) \setminus D(u)$. 
  \end{claim}
  \begin{proof}
    Let $r_0$ and $r_n$ be the endpoints of $P_I$. The vertices of
    $D(u) \symdiff D(v)$ are split into those closest to $r_0$
    and those closest to $r_n$. We first note that every 
    vertex $a \in D(u) \symdiff D(v)$ is compatible with that
    endpoint it is closest to, by considering the $a$-closest pair in
    $R(u,v)$. Next, let $a \in D(u) \setminus D(v)$ be a point closer
    to $r_0$ than $r_n$, and write $P_I=r_0r_1 \ldots r_n$. 
    We claim by induction that $(a, r_i) \in R(u,v)$
    for every $i=0, \ldots, n$. 
    As a base case, we already have $(a, r_0) \in R(u,v)$. 
    The inductive step then goes as follows:
    \begin{enumerate}
    \item Assume $(a, r_i) \in R(u,v)$ for some $i<n$. 
      Then $(a, r_i)$ and $(r_{i+1}, r_{i+1})$ via $\omega_{r_i}$
      produces $(a', r_{i+1})$ where $a'$ is the vertex following $a$
      on $P_u$ and may be $a'=r_0$
    \item $(a,r_i), (a', r_{i+1})$ via $\omega_{r_i}$ produces $(a,r_{i+1})$
    \end{enumerate}
    By induction we then get $(a,r) \in R(u,v)$ for every $r \in V(P_I)$,
    and the rest of the cases of the claim follow by symmetry. 
  \end{proof}
  
  \noindent
  The rest of the proof is now easy. Let $a \in D(u) \setminus D(v)$
  and $b \in D(v)$; the other case is symmetric. If $b$ is a shared
  value, then $(a,b) \in R(u,v)$ by the above. Otherwise, let $r$ be
  an arbitrary shared value; then $(a,r), (r,b) \in R(u,v)$ and
  $\omega_r$ produces $(a,b) \in R(u,v)$. 
  Hence the only excluded pairs are on shared values. 
  Finally, $(r,r) \in R(u,v)$ for every shared value $r$
  by considering the $r$-closest pair. 
\end{proof}

\noindent
We also need the following standard result.

\begin{lemma} \label{lemma:majority}
  Let $\lambda_0$ be a partial labelling on a set of vertices $U
  \subseteq V$. Then there exists a labelling $\lambda \in \opt$
  that extends $\lambda_0$ if and only if 
  $(\lambda_0(u), \lambda_0(v)) \in R(u,v)$ 
  for every pair $u, v \in U$. 
\end{lemma}
\begin{proof}
  Define the operation $m(a,b,c)=((a \lca b) \lcaskew (a \lca c))
  \lcaskew (b \lca c)$ (for an arbitrary choice of root). Then $m$
  preserves $\opt$, and it is readily verified that $m$ is a
  \emph{majority operation}, i.e., $m(x,x,y)=m(x,y,x)=m(y,x,x)x$. 
  It follows that $\opt$ is characterized by its binary projections. 
\end{proof}

\noindent
This gives us the following algorithmic consequence.

\begin{lemma} \label{lemma:unary-support-suffices}  
  There is a labelling $\lambda \in \opt$ such that
  for every variable $v$ with $D_I(v)$ non-empty,
  we have $\lambda(v) \in D_I$.   
\end{lemma}
\begin{proof}
  Let $U = \{v \in V \mid D_I(v) \neq \emptyset\}$ be the set of
  vertices with at least one integral value in the support.
  By Lemma~\ref{lemma:majority}, it suffices to produce a partial
  labelling $\lambda_0: U \to D_I$ such that 
  every binary projection of $\lambda_0$ is supported by $\opt$, i.e.,
  for every pair of distinct vertices $u, v \in U$
  we have $(\lambda_0(u), \lambda_0(v)) \in R(u,v)$.
  For this, define an arbitrary total order on
  $D_I$, and define $\lambda_0$ by selecting for every $v \in U$
  the value of $D_I(v)$ that is earliest according to this order. 
  Then for every pair $u, v \in U$ either $\lambda_0(u)=\lambda_0(v)$
  or one of the values $\lambda_0(u), \lambda_0(v)$ is not shared in
  $R(u,v)$. In both cases, by Lemma~\ref{lemma:all-non-shared}
  we have $(\lambda_0(u), \lambda_0(v)) \in R(u,v)$. 
\end{proof}

\noindent
Let us for reusability spell out the explicit assumptions and
requirements made until now. 

\begin{theorem} \label{theorem:unary-support-suffices}
  Let $I=(G=(V,E), \tau, \mu, q)$ be an instance of \textsc{Zero Extension}
  with no isolated vertices and 
  where every connected component of $G$ contains at least two
  terminals, and where $\mu$ is an induced tree metric for some tree
  $T$ and integral nodes $D_I \subseteq T$. Additionally, assume a
  collection of cost functions ${\mathcal F}=(f_i(\bar v_i))_{i=1}^m$
  has been given, where for every $f_i$ the scope is contained in $V$  
  and where $f_i$ is weakly tree submodular for every rooted version
  of $T$. Let $I'$ be the VCSP instance created from the sum of the 
  cost functions of $I$ and $\mathcal{F}$. Then $I'$ has the domain
  consistency property, i.e., there is an integral relaxed optimum if
  and only if every vertex $v$ is integral in at least one relaxed
  optimum of $I'$.  
\end{theorem}
\begin{proof}
  As noted in Lemma~\ref{lemma:path-support-everywhere}, due to the
  instance $I$ every vertex $v$ is either incident with at least one
  edge of $E$ that has non-zero length in at least one optimum,
  or it holds that $u=v$ in every optimal assignment, where 
  $u$ is such a vertex. Hence $D(v)$ forms a path in $T$,
  and for every pair of variables $u, v \in V$ 
  the conclusion of Lemma~\ref{lemma:all-non-shared} applies,
  even for the projection $R(u,v)$ of $\opt(I')$ (as opposed to just
  $\opt(I)$). The result follows as in
  Lemma~\ref{lemma:unary-support-suffices}. 
\end{proof}

\subsection{Gap algorithms for general induced tree metrics}

\noindent
We now use the results of Theorem~\ref{theorem:unary-support-suffices} to
provide $\FPT$ algorithms parameterized by the gap parameter $k-\rho$.


\thmtrees*
\begin{proof}
  This algorithm is similar to the algorithm for a leaf metric, except
  that we are not as easily able to test whether every variable has an
  integral value in $\opt$. By the results of 
  Section~\ref{section:leaf-metric}, the value of $\opt$ is witnessed
  by the collection of min-cuts for edges in $T$; we will use this as
  a value oracle for $I$. We initially compute a max-flow across every
  edge of $T$, then for every assignment made we can compute the new
  value of $\opt$ using $O(|T|)$ calls to augmenting path algorithms. 
  This allows us to test for optimality of an assignment 
  in $O(|T|m)$ time. The branching step then in general iterates over
  at most $n$ variables, testing at most $|D|$ assigned values for
  each, and testing for optimality each time. Hence the local work in
  a single node of the branching tree is $O(|T||D|nm)$. 
  This either produces a variable for branching on or (by
  Theorem~\ref{theorem:unary-support-suffices}) produces an integral
  assignment, and in each branching step the value of $\rho$ increases
  but $q$ does not. The time for the initial max-flow computation is
  eaten by the factor $|T|nm$. The result follows. 
\end{proof}

\noindent
For \textsc{Metric Labelling}, we first need to restrict the unary
cost functions to be weakly tree-submodular. 

\begin{lemma}
  Let $f \colon V(T) \to \R$ be a unary function on a tree $T$. 
  Then $f$ is weakly tree submodular on $T$ for every choice of 
  root $r \in V_T$ if and only if it observes the following
  \emph{interpolation property}: for any nodes $u, v \in V(T)$,
  at distance $\dist_T(u,v)=d$, and every $i \in [d-1]$,
  let $w_i$ be the node on $T[u,v]$ satisfying $d_T(u,w_i)=i$.
  Then for any such choice of $u$, $v$ and $i$, 
  it holds that $f(w_i) \leq ((d-i)/d)f(u) + (i/d)f(v)$.
\end{lemma}
\begin{proof}
  On the one hand, assume that $f$ has the gradient property,
  and let $u, v \in V(T)$ and let $r \in V(T)$ be a choice of root. 
  We need to show that $f$ is preserved by the $r$-rooted weakly tree
  submodular multimorphism. If $r$ does not lie on $T[u,v]$,
  or $r \in \{u,v\}$, then this is vacuous. Otherwise, let $\lca(u,v)$
  be the $i$:th node of $T[u,v]$, $i \in [d-1]$. 
  Then we need to show
  \[
  f(u) + f(v) \geq f(w_i) + f(w_{d-i}),
  \]
  which clearly holds. On the other hand, let $u, v \in V(T)$ 
  and $i \in [\dist_T(u,v)-t]$ be such that the interpolation
  property does not hold. Let $w_a$, $a<i$ be the last vertex 
  before $w_i$ on $T[u,v]$ such that the interpolation equality
  holds for $u$, $v$ and $a$, and let $w_b$, $b>i$ be the first 
  vertex after $w_i$ for which it holds. These vertices clearly 
  exist, possibly with choices $w_a=u$ and $w_b=v$. 
  Then $w_a$ and $w_b$ are a witness that the $w_i$-rooted weak tree
  submodularity multimorphism is not a multimorphism of $f$.
\end{proof}

\noindent
In particular, let $f_0 \colon U \to \Z^+$ be a non-negative function
defined on a subset $U$ of the nodes of a tree $T$, and say that $f_0$ 
\emph{admits an interpolation on $T$} if there is an extension 
$f \colon V(T) \to \Z^+$ with the interpolation property
such that $f(v)=f_0(v)$ for every $v \in U$. 
In particular, if $U$ is the set of leaves of $T$, then every 
function $f_0$ admits an interpolation by simply padding with zero
values (although stronger interpolations are in general both possible
and desirable). 

\thmtreesmetric*
\begin{proof}
  Assume that $G$ is connected, or else repeat the below for every
  connected component of $G$. Select two arbitrary vertices $u, v \in V$
  and exhaustively guess their labels; in the case that you guess them
  to have the same label, identify the vertices in $I$ (adding up
  their costs in $\sigma$) and select a new pair to guess on. 
  Note that this takes at most $O(|D|^2n)$ time, terminating 
  whenever you have guessed more than one label in a branch or
  when you have guessed that all vertices are to be identical.
  This guessing phase can only increase the value of $\rho$. 
  We may now treat $u$ and $v$ as terminals, and the instance
  $I$ as the sum of a \textsc{Zero Extension} instance on those two
  terminals and a collection of additional unary cost functions
  $\sigma(v',\cdot)$, as in Theorem~\ref{theorem:unary-support-suffices}.
  Note that the resulting VCSP is tractable, i.e., the value of an
  optimal solution can be computed in polynomial time. 
  The running time from this point on consists of iterating through
  all variables verifying whether each one has an integral value in
  some optimal assignment, and branching exhaustively on its value if
  not.  
\end{proof}

\noindent
In particular, as noted, for a leaf metric $\mu$ the algorithm applies
without any assumptions on $\sigma$ (and without $T$ being explicitly
provided).

\section{Conclusions}
\label{section:conclusions}
\noindent
We have given a range of algorithmic results for the
\textsc{Zero Extension} and \textsc{Metric Labelling} problems from 
a perspective of parameterized complexity.

Most generally, we showed that \textsc{Zero Extension} is $\FPT$
parameterized by the number of \emph{crossing edges} of an optimal
solution, \ie the number of edges whose endpoints receive distinct
labels, for a very general class of cost functions $\mu$ that need not
even be metrics. 
This is a relatively straight-forward application of the technique of
recursive understanding~\cite{RandomContraction}.

For the more reasonable case that $\mu$ is a metric, \ie observing the
triangle inequality, we gave two stronger results for the same
parameter. First, we showed a linear-time $\FPT$ algorithm, with a better 
parameter dependency, using an important separators-based
algorithm. Second, and highly surprisingly, we show that 
every graph $G$ with a terminal set $S$
admits a polynomial-time computable, polynomial-sized 
\emph{metric sparsifier} $G'$, with $O(k^{s+1})$ edges, 
such that $(G',S)$ mimics the behaviour of $(G,S)$ over \emph{any}
metric on at most $s$ labels (up to solutions with crossing number
$k$). This is a direct and seemingly far-reaching generalization of
the polynomial kernel for \textsc{$s$-Multiway Cut}~\cite{KratschW12FOCS}, 
which corresponds to the special case of the uniform metric. 

Finally, we further developed the toolkit of discrete relaxations to
design $\FPT$ algorithms under a \emph{gap parameter}
for \Problem{Zero Extension} and \Problem{Metric Labelling}
where the metric is an induced tree metric (\ie a restriction
of a tree metric to a subset of the values).
This in particular involves a more general $\FPT$ algorithm approach,
supported by an applicability condition of \emph{domain consistency},
relaxing the previously used persistence condition. 

Let us highlight some questions. 
First, is there a lower bound on the size of a metric sparsifier 
for $s$ labels for \textsc{Zero Extension}? This is particularly 
relevant since the existence of a polynomial kernel for
\textsc{$s$-Multiway Cut} whose degree does not scale with $s$
is an important open problem, and since the metric sparsifier is a
more general result. 

Second, can the $\FPT$ algorithms for induced tree metrics parameterized
by the relaxation gap be generalised to restrictions of other tractable metrics, 
such as graph metrics for  median graphs or the most general tractable
class of orientable modular graphs~\cite{Hirai13SODA}? 
Complementing this, what are the strongest possible gap parameters
that allow $\FPT$ algorithms for metrics that are either arbitrary,
or do not explicitly provide their relaxation?

More broadly, we also ask how far the method of discrete relaxations 
stretches in general. 
Let \textsc{Subset VCSP} be the class of problems defined 
by a (presumably tractable) integer-valued language $\Gamma$ 
on a domain $D$, and a subset $D_I \subseteq D$ of integral values in the domain,
where the problem is equivalent to the VCSP on $\Gamma$ restricted to
the domain $D_I$, but where we ask for an $\FPT$ algorithm 
parameterized by the corresponding relaxation gap. 
Can it be characterized for which $\Gamma$ and $D_I$ 
this problem is $\FPT$, and/or for which cases  the domain consistency
property holds?

\bibliographystyle{abbrv}
\bibliography{biblio}

\end{document}